\documentclass[pldi,preprint,nocopyrightspace]{sigplanconf-pldi16}

\usepackage{srcltx}
\usepackage{url}

\usepackage{latexsym,amsthm,amsmath,amssymb,stmaryrd}
\usepackage[toc,page]{appendix}

\usepackage{xspace}
\usepackage{enumitem}
\usepackage{graphics}
\usepackage{listings}
\usepackage{supertabular}
\usepackage{marvosym}
\usepackage{mathpartir}
\usepackage{color,soul} 
\usepackage[table]{xcolor} 
\usepackage[multiuser=true,draft,silent]{fixme}
\usepackage{hyperref}

\usepackage{listings}
\newcommand{\ruby}[1]{\textsf{#1}}

\usepackage{amsthm}
\newtheorem{lemma}{Lemma}
\newtheorem{definition}{Definition}
\newtheorem{theorem}{Theorem}

\definecolor{lightgreen}{rgb}{.85,.95,.85}
\definecolor{lightblue}{rgb}{.85,.90,1}
\definecolor{lightred}{rgb}{.95,.85,.85}
\definecolor{lightgrey}{rgb}{.95,.95,.95}
\sethlcolor{lightred}

\definecolor{drkyellow}{rgb}{0.4,0.3,0}
\definecolor{drkred}{rgb}{0.5,0,0}
\definecolor{drkgreen}{rgb}{0,0.5,0}

\definecolor{dkgreen}{rgb}{0.1,0.40,0.1}

\definecolor{dkred}{rgb}{0.5,0,0}
\definecolor{medred}{rgb}{0.6,0,0}
\definecolor{gray}{rgb}{0.5,0.5,0.5}
\definecolor{mauve}{rgb}{0.58,0,0.82}

\newcommand{\DE}{\ensuremath{\mathit{E}}}
\newcommand{\DT}{\ensuremath{\mathit{DT}}}

\newcommand{\TE}{\ensuremath{\Gamma}}
\newcommand{\TEalt}{\ensuremath{\Delta}}
\newcommand{\TT}{\ensuremath{\mathit{TT}}}

\newcommand{\DS}{\ensuremath{\mathit{S}}}
\newcommand{\TS}{\ensuremath{\mathit{TS}}}
\newcommand{\config}[1]{\langle{}#1\rangle{}}
\newcommand{\rname}[1]{\rm \text{(#1)}}
\newcommand{\dom}{\mathop\textit{dom}}
\newcommand{\tselt}[4]{(#1[#2], \config{#3, #4})}
\newcommand{\cache}{\ensuremath{\mathit{X}}}
\newcommand{\invcache}[2]{#1 \backslash #2}
\newcommand{\upgcache}[2]{#1[#2]}
\newcommand{\deriv}{\ensuremath\mathcal{D}}
\newcommand{\derivm}{\deriv_M}
\newcommand{\derivt}{\deriv_{\leq}}
\newcommand{\blame}{\textit{blame}}

\newcommand{\snil}{\mathsf{nil}}
\newcommand{\sself}{\mathsf{self}}
\newcommand{\snew}[1]{#1.\mathsf{new}}
\newcommand{\sif}[3]{\mathsf{if}~#1~\mathsf{then}~#2~\mathsf{else}~#3}

\newcommand{\sprem}[2]{\ensuremath{\lambda #1.#2}}
\newcommand{\sdef}[2]{\mathsf{def}~#1 = #2}
\newcommand{\smethsig}[2]{\mathsf{type}~#1:#2}

\newcommand{\stype}[1]{\mathsf{type\_of}(#1)}

\newcommand{\pruby}{$\mathcal{P}$Ruby}

\newcommand{\name}{Hummingbird}

\newcommand{\sname}[1]{\textit{#1}}

\fxuselayouts{nomargin,footnote}
\fxusetheme{color}
\FXRegisterAuthor{bree}{abree}{\color{blue}BR}
\FXRegisterAuthor{jeff}{ajeff}{\color{red}JSF}

\title{Just-in-Time Static Type Checking for Dynamic Languages}

\authorinfo{Brianna M. Ren \and Jeffrey S. Foster}
          {University of Maryland, College Park}
          {\{bren, jfoster\}@cs.umd.edu}

\begin{document}

\maketitle

\begin{abstract}
  Dynamic languages such as Ruby, Python, and JavaScript have many compelling benefits, but the lack of
  static types means subtle errors can remain latent in code for a
  long time. While many researchers have developed various systems to
  bring some of the benefits of static types to dynamic languages,
  prior approaches have trouble dealing with metaprogramming,
  which generates code as the program executes.  In this paper, we
  propose Hummingbird, a new system that uses a novel technique,
  \emph{just-in-time static type checking}, to type check Ruby code
  even in the presence of metaprogramming. In Hummingbird,
  method type signatures are gathered dynamically at run-time, as
  those methods are created. When a method is called, Hummingbird
  statically type checks the method body against current type
  signatures. Thus, Hummingbird provides thorough static checks on a
  per-method basis, while also allowing arbitrarily complex
  metaprogramming. For performance, Hummingbird memoizes the static type
  checking pass, invalidating cached checks only if necessary.
  We formalize Hummingbird using a core, Ruby-like language and prove it
  sound. To evaluate Hummingbird, we applied it to six apps,
  including three that use Ruby on Rails, a powerful framework that
  relies heavily on metaprogramming.
  We found that all apps typecheck
  successfully using Hummingbird, and that Hummingbird's performance
  overhead is reasonable. We applied
  Hummingbird to earlier versions
  of one Rails app and found several type errors that had been introduced and
  then fixed. Lastly, we demonstrate using Hummingbird in Rails
  development mode to typecheck an app as live updates are applied to it.
\end{abstract}


\category{F.3.2}{Semantics of Programming Languages}
{Program analysis}

\keywords 
type checking; dynamic languages; Ruby


\section{Introduction}


Many researchers have explored ways to bring the benefits of static
typing to dynamic
languages~\cite{furr:oops09,furr:oopsla09,an:popl11,ren:oops13,aycock:python,anderson:javascript,thiemann:javascript,Lerner:2013:TRT:2508168.2508170,Lerner:2013:CFF:2524984.2524990,bierman2014understanding,Rastogi:2015:SEG:2676726.2676971,Tobin-Hochstadt:2008:DIT:1328438.1328486,Maidl:2014:TLO:2617548.2617553,ancona:rpython}. However,
many of these prior systems do not work well in the presence of
\emph{metaprogramming}, in which code the program relies on is
generated as the program executes. The challenge is that purely static
systems cannot analyze metaprogramming code, which is often
complicated and convoluted; and prior mixed static/dynamic systems are
either cumbersome or make certain limiting assumptions.
(Section~\ref{sec:related} discusses prior work in detail.)


In this paper, we introduce \name{},\footnote{A hummingbird can
  dynamically flap its wings while statically hovering in place.} a
type checking system for Ruby that solves this problem using a new
approach we call \emph{just-in-time static type checking}.
In \name{}, user-provided type annotations actually execute at
\emph{run-time}, adding types to an environment that is maintained
during execution.  As metaprogramming code creates new methods, it, too,
executes type annotations to assign types to dynamically created
methods. Then whenever a method \ruby{m} is called, \name{}
\emph{statically} type checks \ruby{m}'s body in the current dynamic
type environment. More precisely, \name{} checks that \ruby{m} calls
methods at their types as annotated, and that \ruby{m} itself matches
its annotated type.  Moreover, \name{} caches the type check so that
it need not recheck~\ruby{m} at the next call unless the dynamic type
environment has changed in a way that affects \ruby{m}.


Just-in-time static type checking provides a highly effective tradeoff between purely
dynamic and purely static type checking. On the one hand,
metaprogramming code is very challenging to analyze statically, but in
our experience it is easy to create type annotations at run time for
generated code. On the other hand, by statically analyzing whole
method bodies, we catch type errors earlier than a purely
dynamic system, and we can soundly reason about all possible execution paths
within type checked methods. (Section~\ref{sec:overview} shows how several
examples of metaprogramming are handled by \name{}.)

To ensure our approach to type checking is correct, we formalize
\name{} using a core, Ruby-like language in which
method creation and method type annotation can
occur at arbitrary points during execution. We provide a
flow-sensitive type checking system and a dynamic semantics that
invokes the type system at method entry, caching the resulting typing
proof. Portions of the cache may be invalidated as new methods are
defined or type annotations are changed. We prove
soundness for our type
system. (Section~\ref{sec:formalism} presents the formalism.)

Our implementation of \name{} piggybacks on two prior systems we developed. We use
the Ruby Intermediate Language~\cite{furr:dls09-ril,DRuby} to parse input
Ruby files and translate them to simplified control-flow graphs. We
use RDL~\cite{strickland:dls12,RDL}, a Ruby contract system, to intercept method calls
and to represent and store method type signatures at run
time. \name{} supports an extensive set of typing features, including
union types, intersection types, code blocks (anonymous functions),
generics, modules, and type casts, among others.
(Section~\ref{sec:implementation} describes our implementation.)

We evaluated \name{} by applying it to six Ruby apps. Three use Ruby
on Rails (just ``Rails'' from now on), a popular, sophisticated web
app framework that uses metaprogramming heavily both to make Rails
code more compact and expressive and to support ``convention over
configuration.'' We should emphasize that Rails's use of
metaprogramming makes static analysis of it very
challenging~\cite{an:ase09}.  Two additional apps use other styles of
metaprogramming, and the last app does not use metaprogramming, as a
baseline.

We found that all of our subject apps
type check successfully using \name{}, and that dynamically
generated types are essential for the apps that use metaprogramming.
We also found that \name{}'s performance overhead
ranges from 19\% to 469\%,
which is much better than prior
approaches \cite{an:popl11,ren:oops13},
and that caching is essential to achieving this performance.
For one Rails app, we ran type checking on many prior versions,
and we found a total of six type errors that had been introduced and
then later fixed. We also ran the app in Rails \emph{development
  mode}, which reloads files as they are edited, to demonstrate how
\name{} type check caching behaves in the presence of modified methods.
(Section~\ref{sec:experiments} reports on our results.)

In summary, we believe \name{} is an important step forward in
our ability to bring the benefits of static typing to dynamic
languages while still supporting flexible and powerful metaprogramming
features.

\section{Overview}
\label{sec:overview}

We begin our presentation by showing some uses of metaprogramming in Ruby and the
corresponding \name{} type checking process. The examples below
are from the experiments in Section~\ref{sec:experiments}.

\begin{figure}
\begin{lstlisting}
class Talk < ActiveRecord::Base
  belongs_to :owner, :class_name => "User" ##\label{line:belongs_to}#
  ....
  type :owner?, "(User) -> %bool"
  def owner?(user) ##\label{line:owner}#
    return owner == user
end end

module ActiveRecord::Associations::ClassMethods ##\label{line:assoc-start}#
  pre(:belongs_to) do |*args|
     hmi = args[0]
     options = args[1]
     hm = hmi.to_s
     cn = options[:class_name] if options
     hmu = cn ? cn : hm.singularize.camelize
     type hm.singularize, "() -> #{hmu}"
     type "#{hm.singularize}=", "(#{hmu}) -> #{hmu}"
     true
end end ##\label{line:assoc-end}#
\end{lstlisting}
  \nocaptionrule{}
  \caption{Ruby on Rails Metaprogramming.}
  \label{fig:belongs_to}
\end{figure}

\paragraph*{Rails Associations.}

The top of Figure~\ref{fig:belongs_to} shows an excerpt from the
\sname{Talks} Rails app. This code defines a class \ruby{Talk} that is
a \emph{model} in Rails, meaning an instance of \ruby{Talk} represents
a row in the \ruby{talks} database table. The change in case and
pluralization here is not an accident---Rails favors ``convention over
configuration,'' meaning many relationships that would otherwise be
specified via configuration are instead implicitly expressed by using similar or the
same name for things.

In this app, every talk is owned by a user, which in implementation
terms means a \ruby{Talk} instance has a foreign key \ruby{owner\_id}
indicating the owner, which is an instance of class \ruby{User} (not
shown). The existence of that relationship is defined on
line~\ref{line:belongs_to}. Here it may look like \ruby{belongs\_to}
is a keyword, but in fact it is simply a method call. The call
passes the \emph{symbol} (an interned string) \ruby{:owner} as the first
argument, and the second argument is a hash that maps symbol
\ruby{:class\_name} to string \ruby{"User"}.

Now consider the \ruby{owner?} method, defined on
line~\ref{line:owner}. Just before the method, we introduce a type
annotation indicating the method takes a \ruby{User} and
returns a boolean. Given such an annotation, \name{}'s goal is
to check whether the method body has the indicated type.\footnote{In
  practice \ruby{type} takes another argument to tell \name{} to type
  check the body, in contrast to library and framework methods whose
  types are trusted. We elide this detail for simplicity.}  This
should be quite simple in this case, as the body of \ruby{owner?}
just calls
no-argument method \ruby{owner} and checks whether the result is equal
to \ruby{user}.

However, if we examine the remaining code of \ruby{Talk} (not shown),
we discover that \ruby{owner} is not defined anywhere in the class!
Instead, this method is created at run-time by
\ruby{belongs\_to}. More specifically, when \ruby{belongs\_to} is
called, it defines several convenience methods that perform
appropriate SQL queries for the
relationship~\cite{Associations}, in this case to get the
\ruby{User} instance associated with the \ruby{Talk}'s owner.
Thus, as we can see, it is critical for \name{} to handle such dynamically
created methods even to type check simple Rails code.

Our solution is to instrument \ruby{belongs\_to} so that, just as it
creates a method dynamically, it also creates method type signatures
dynamically. The code on
lines~\ref{line:assoc-start}--\ref{line:assoc-end} of
Figure~\ref{fig:belongs_to} accomplishes this.
\name{} is built on RDL, a Ruby
contract system for specifying pre- and postconditions~\cite{strickland:dls12,RDL}.
The precondition is specified via a \emph{code block}---an anonymous
function (i.e., a lambda) delimited by
\ruby{do}$\ldots$\ruby{end}---passed to \ruby{pre}.
Here the code block trivially returns \ruby{true} so the precondition is always
satisfied (last line) and, as a side effect,
creates method type annotations for \ruby{belongs\_to}.

In more detail, \ruby{hmi} is set to the first argument to
\ruby{belongs\_to}, and \ruby{options} is either \ruby{nil} or the
hash argument, if present. (Here \ruby{hm} is shorthand for ``has
many,'' i.e., since the \ruby{Talk} belongs to a \ruby{User}, the
\ruby{User} has many \ruby{Talk}s.) Then \ruby{hmu} is set to either
the \ruby{class\_name} argument, if present, or
\ruby{hmi} after singularizing and camel-casing
it. Then \ruby{type} is called twice, once to give a type to a
getter method created by \ruby{belongs\_to}, and once for a setter
method (whose name ends with \ruby{=}). Notation \ruby{\#\{e\}}
inside a string evaluates the expression \ruby{e} and inserts the
result in the string.
In this particular case, these
two calls to \ruby{type} evaluate to
\begin{lstlisting}[numbers=none]
  type "owner", "() -> User"
  type "owner=", "(User) -> User"
\end{lstlisting}

Now consider executing this code. When \ruby{Talk} is loaded,
\ruby{belongs\_to} will be invoked, adding those type signatures to
the class. Then when \ruby{owner?} is called, \name{} will perform
type checking using currently available type information, and so it
will be able to successfully type check the body. Moreover, notice
this approach is very flexible. Rails does not require that
\ruby{belongs\_to} be used at the beginning of a class or even
during this particular class definition. (In Ruby, it is possible to
``re-open'' a class later on and add more methods to it.) But no
matter where the call occurs, it must be before \ruby{owner?} is
called so that \ruby{owner} is defined. Thus in this case, \name{}'s
typing strategy matches well with Ruby's semantics.


\paragraph*{Type Checking Dynamically Created Methods.}

\begin{figure}
\begin{lstlisting}
module Rolify::Dynamic
  def define_dynamic_method(role_name, resource)
    class_eval do
      define_method("is_#{role_name}?".to_sym) do ##\label{line:block}#
        has_role?("#{role_name}")
      end if !method_defined?("is_#{role_name}?".to_sym) ##\label{line:postif}#
      ...
  end end

  pre :define_dynamic_method do |role_name, resource| ##\label{line:define-pre}#
    type "is_#{role_name}?", "() -> %bool"
    true
end end

class User; include Rolify::Dynamic end ##\label{line:rolify}#
user = User.first
user.define_dynamic_method("professor", ...)
user.define_dynamic_method("student", ...)
user.is_professor?
user.is_student?
\end{lstlisting}
  \nocaptionrule{}
  \caption{Methods Dynamically Created by User Code.}
  \label{fig:define_method}
\end{figure}

In the previous example, we trusted Rails to dynamically
generate code matching the given type signature.
Figure~\ref{fig:define_method} shows an example, extracted
from \sname{Rolify}, in which user code dynamically generates a method. The first
part of the figure defines a module (aka mixin) with a two-argument method
\ruby{define\_dynamic\_method}. The method body calls
\ruby{define\_method} to create a method named using the first
argument, as long as that method does not exist (note the postfix
\ruby{if} on line~\ref{line:postif}). Similarly to earlier,
line~\ref{line:define-pre} adds a precondition
to \ruby{define\_dynamic\_method} that provides an appropriate method
type. (We do not check for a previous type definition since adding the
same type again is harmless.)

The code starting at line~\ref{line:rolify} uses the module. 
This particular code is not from our experiment but is merely for
expository purposes.
Here we
(re)open class \ruby{User} and mix in the module. Then we create a
user; call \ruby{define\_dynamic\_method} twice; and then call the
generated methods \ruby{is\_professor?} and \ruby{is\_student?}.

In this case, since the generated methods have type annotations and
are in user code, \name{} type checks their bodies when
they are called, just like any other user-defined method with a type. For example, consider the call to
\ruby{is\_professor?}, which is given type \ruby{()
  $\rightarrow$ \%bool}. At the call, \name{} type checks the code block
at line~\ref{line:block} and
determines that it has no arguments and that its body returns a
boolean, i.e., it type checks.


\paragraph*{User-provided Type Signatures.}

\begin{figure}
\begin{lstlisting}
Transaction = Struct.new(:type, :account_name, :amount) ##\label{line:struct}#
class ApplicationRunner
  def process_transactions
    @transactions.each do |t|
      name = t.account_name
      ...
  end ... end
  field_type :@transactions, "Array<Transaction>" ##\label{line:field_type}#
end

class Struct
  def self.add_types(*types)
    members.zip(types).each {|name, t|
      self.class_eval do 
        type name, "() -> #{t}"
        type "#{name}=", "(t) -> #{t}"
      end
    }
  end
end
Transaction.add_types("String", "String", "String")
\end{lstlisting}
  \nocaptionrule{}
  \caption{Type Signatures for Struct.}
  \label{fig:struct}
\end{figure}

In the examples so far, the types for dynamically created methods
could be determined automatically. However, consider
Figure~\ref{fig:struct}, which shows an excerpt from \sname{CCT} that
uses \ruby{Struct} from the Ruby core library.
Line~\ref{line:struct} creates a new class, instances of which are
defined to have getters \ruby{type}, \ruby{account\_name}, and
\ruby{amount}, and setters \ruby{type=}, \ruby{account\_name=}, and
\ruby{amount=}.  The \ruby{process\_transactions} method iterates
through instance field \ruby{@transactions} (whose type is provided on
line~\ref{line:field_type}), and calls the \ruby{account\_name} method
of each one.

From line~\ref{line:struct} we know the \ruby{account\_name} method
exists, but we do not know its type. Indeed, a ``struct field'' can
hold any type by default. Thus, to fully type check the body of
\ruby{process\_transactions}, we need more information from the
programmer to specify the type of \ruby{account\_name}.

The bottom part of Figure~\ref{fig:struct} defines a new method,
\ruby{add\_types}, that the programmer can call to indicate desired
struct field types. The types are given in the same order as the
constructor arguments, and the body of \ruby{add\_types} uses
\ruby{zip} to pair up the constructor arguments (retrieved via
\ruby{members}) and the types, and then iterates through the pairs,
creating the appropriate type signatures for the getters and
setters. The last line of the figure uses \ruby{add\_types} to create
type signatures for this example, allowing us to type
check \ruby{process\_transactions} when it is called.

In this particular case, we could have individually
specified type signatures for the methods of
\ruby{Transaction}. However, because \name{} lets programmers write
arbitrary Ruby programs to generate types, we were able to develop
this much more elegant solution.

\section{Formalism}
\label{sec:formalism}

\begin{figure}[t!]
  \small
  \begin{displaymath}
    \begin{array}{lrcll}
      \text{values} & v & ::= & \snil \mid [A] \\
      \text{expressions} & e & ::= & v
                                     \mid x \mid \sself
                                     \mid x = e
                                     \mid e; e
                                     \mid \snew{A} \\
      && \mid & \sif{e}{e}{e}
                \mid e.m(e) \\
      && \mid & \sdef{A.m}{b} \mid \smethsig{A.m}{\tau_m} \\
      \text{premths} & b & ::= & \sprem{x}{e} \\
      \text{val typs} & \tau & ::= &
      A \mid \snil \\
      \text{mth typs} & \tau_m & ::= & \tau \rightarrow \tau \\ \\
    \end{array}
  \end{displaymath}

  \begin{displaymath}
      \begin{array}{r@{\;}c@{\;}lr@{\;}c@{\;}lr@{\;}c@{\;}l}
        x & \in & \text{var ids} &
        m & \in & \text{mth ids} &
        A & \in & \text{cls ids} \\
      \end{array}
    \end{displaymath}

    \begin{displaymath}
      \begin{array}{l@{\;}rcl}
        \text{dyn env} & \DE & : & \text{var ids} \rightarrow \text{vals} \\
        \text{dyn cls tab} & \DT & : & \text{cls ids} \rightarrow \text{mth ids}
                                      \rightarrow \text{premths} \\
        \text{contexts} & C & ::= & \Box
                                    \mid x = C
                                    \mid C.m(e)
                                    \mid v.m(C) \\
                       && \mid & C; e
                                 \mid \sif{C}{e}{e} \\
        \text{stack} & \DS & ::= & \cdot \mid (\DE, C)::\DS \\
        \text{type env} & \TE, \TEalt & : & \text{var ids} \rightarrow \text{val
                     typs} \\
        \text{type tab} & \TT & : & \text{cls ids} \rightarrow \text{mth ids}
                  \rightarrow \text{mth typs} \\
        \text{cache} & \cache & ::= & \text{cls ids} \rightarrow \text{mth ids} \rightarrow \derivm \times \derivt \\
        \text{typ chk deriv} & \derivm & ::= & \TT \vdash \config{\TE, e}
      \Rightarrow \config{\TE', \tau} \\
        \text{subtyp deriv} & \derivt & ::= & \tau_1 \leq \tau_2 \\
    \end{array}
  \end{displaymath}

  \nocaptionrule{}
  \caption{Source Language and Auxiliary Definitions.}
  \label{fig:source-lang}
\end{figure}

We formalize \name{} using the core, Ruby-like language shown at the
top of Figure~\ref{fig:source-lang}. \emph{Values}~$v$ include
$\snil$, which can be treated as if it has any type, and $[A]$, which
is an instance of class $A$. Note that we omit both fields and
inheritance from our formalism for simplicity, but they are handled by
our implementation.

\begin{figure*}[t!]
  \small

  \framebox{$\TT \vdash \config{\TE, e} \Rightarrow \config{\TE', \tau}$}
  \begin{displaymath}
    \begin{array}{c}
      \infer[\rname{TNil}]
      { }
      {\TT \vdash \config{\TE, \snil} \Rightarrow \config{\TE, \snil}}
      \qquad

      \infer[\rname{TObject}]
      { }
      {\TT \vdash \config{\TE, [A]} \Rightarrow \config{\TE, A}}
      \qquad

      \infer[\rname{TSelf}]
      { }
      {\TT \vdash \config{\TE, \sself} \Rightarrow \config{\TE, \TE(\sself)}}
      \qquad

      \infer[\rname{TVar}]
      { }
      {\TT \vdash \config{\TE, x} \Rightarrow \config{\TE, \TE(x)}}
      \\ \\

      \infer[\rname{TSeq}]
      {\TT \vdash \config{\TE, e_1} \Rightarrow \config{\TE_1, \tau_1} \\\\
      \TT \vdash \config{\TE_1, e_2} \Rightarrow \config{\TE_2, \tau_2}}
      {\TT \vdash \config{\TE, e_1; e_2} \Rightarrow \config{\TE_2, \tau_2}}

      \quad

      \begin{array}{c}
        \infer[\rname{TAssn}]
        {\TT \vdash \config{\TE, e} \Rightarrow \config{\TE', \tau}}
        {\TT \vdash \config{\TE, x = e} \Rightarrow \config{\TE'[x\mapsto \tau], \tau}}
        \\ \\

        \infer[\rname{TNew}]
        { }
        {\TT \vdash \config{\TE, \snew{A}} \Rightarrow \config{\TE, A}}
      \end{array}

      \quad

      \begin{array}{c}
        \infer[\rname{TDef}]
        { }
        {\TT \vdash \config{\TE, \sdef{A.m}{\sprem{x}{e}}} \Rightarrow \config{\TE, \snil}}
        \\ \\

        \infer[\rname{TType}]
        { }
        {\TT \vdash \config{\TE, \smethsig{A.m}{\tau_m}} \Rightarrow \config{\TE, \snil}}
      \end{array}

      \\ \\

      \infer[\rname{TApp}]
      {\TT \vdash \config{\TE, e_0} \Rightarrow \config{\TE_0, A} \\\\
      \TT \vdash \config{\TE_0, e_1} \Rightarrow \config{\TE_1, \tau} \\\\
      \TT(A.m) = \tau_1\rightarrow \tau_2 \\
      \tau \leq \tau_1
      }
      {\TT \vdash \config{\TE, e_0.m(e_1)} \Rightarrow \config{\TE_1, \tau_2}}

      \quad

      \infer[\rname{TIf}]
      {\TT \vdash \config{\TE, e_0} \Rightarrow \config{\TE', \tau} \\
      \TT \vdash \config{\TE', e_1} \Rightarrow \config{\TE_1, \tau_1} \\\\
      \TT \vdash \config{\TE', e_2} \Rightarrow \config{\TE_2, \tau_2}}
      {\TT \vdash \config{\TE, \sif{e_0}{e_1}{e_2}} \Rightarrow \config{\TE_1 \sqcup \TE_2, \tau_1 \sqcup \tau_2}}

    \end{array}
  \end{displaymath}
  
  \nocaptionrule{}
  \caption{Type Checking System.}
  \label{fig:types}
\end{figure*}

\emph{Expressions}~$e$ include values, variables~$x$, the special
variable~$\sself$, assignments $x=e$, and sequencing $e;e$. Objects
are created with $\snew{A}$. Conditional
$\sif{e_1}{e_2}{e_3}$ evaluates to $e_2$ unless $e_1$ evaluates to
$\snil$, in which case it evaluates to $e_3$. Method invocation
$e_1.m(e_2)$ is standard, invoking the $m$ method based on the run-time type
of $e_1$.

Expression $\sdef{A.m}{\sprem{x}{e}}$, defines
method $m$ of class $A$ as taking argument $x$ and returning $e$. (We
refer to $\sprem{x}{e}$ as a \emph{premethod}.) This form allows
methods to be defined anywhere during execution, thus it combines the
features of Ruby's \ruby{def} and \ruby{define\_method}.
As in Ruby, if $A.m$ is already defined, \ruby{def} overwrites
the previous definition.
The \ruby{def} expression itself evaluates to $\snil$.

Finally, expression
$\smethsig{A.m}{\tau\rightarrow\tau'}$ asserts that method $m$ of class
$A$ has domain type $\tau$ and range type $\tau'$. Types may be
either classes $A$ or $\snil$, the type of expression $\snil$. The
\ruby{type} expression overwrites the previous type of $A.m$, if any.
Like \name{}, there is no ordering dependency between \ruby{def} and
\ruby{type}---the only requirement is that a method's type must be
declared by the time the method is called. The \ruby{type} expression
itself evaluates to $\snil$.

\paragraph*{Type Checking.}

Figure~\ref{fig:types} gives the static type checking rules. As in \name{},
static type checking is performed at run time at method entry---thus
these rules will be invoked as a subroutine by the dynamic semantics (below). The
bottom part of Figure~\ref{fig:source-lang} defines the sets and maps
used in this figure and in the dynamic semantics.

In these rules,
$\TT$ is a \emph{type table} mapping class and method ids $A.m$ to
their corresponding types, as declared by \ruby{type}, and $\TE$ is a
\emph{type environment} mapping local variables to their types. These
rules prove judgments of the form
$\TT \vdash \config{\TE, e} \Rightarrow \config{\TE', \tau}$, meaning
with type table $\TT$, in type environment $\TE$, expression $e$ has
type $\tau$, and after evaluating $e$, the new type environment is
$\TE'$. Using an ``output'' type environment $\TE'$ allows us to
build a flow-sensitive type system, in which variables' types can change
at assignments. Note there is no output $\TT$ because the type table
does not change during static type checking---it only changes as the
program is executed by the dynamic semantics.

The type
rules are largely standard. \rname{TNil} and \rname{TObject} give
$\snil$ and instances the obvious types. \rname{TSelf} and
\rname{TVar} give $\sself$ and local variables their types according
to the type environment. Since none of these four expressions updates
the state, the output type environment is the
same as the input environment.
 
\rname{TSeq} types sequencing, threading the type environment from the
output of $e_1$ to the input of $e_2$. \rname{TAssn} types an
assignment, updating the output type environment to bind the assigned
variable $x$ to the type of the right-hand side. 
\rname{TNew} types
object creation in the obvious way. \rname{TDef} trivially type checks
method definitions. Notice we do not type check the method body; that
will happen at run time when the method is actually
called. \rname{TType} type checks a \ruby{type} expression, which has
no effect during type checking. Such expressions are only evaluated at
run-time, when they update the type table (see below).

One consequence of \rname{TType} is that our type system forbids
typing a method and then immediately calling it in the same method
body. For example, the following method body would fail to type check:
\begin{lstlisting}
def A.m = $\lambda$x.
  def B.m = ...;   # define B.m
  type B.m : ...;  # give B.m a type
  B.new.m          # type error! B.m not in type table
\end{lstlisting}
Here we type check \ruby{A.m}'s body at the first call to it, so the
\ruby{type} expression has not been run---and hence has not bound a
type to \ruby{B.m}---yet. Thus it is a type error to invoke \ruby{B.m}
in the method body.

\begin{figure*}[t!]
  \small

  \framebox{$\config{\cache, \TT, \DT, \DE, e, \DS} \rightarrow \config{\cache', \TT', \DT', \DE', e', \DS'}$}
  \begin{displaymath}
    \begin{array}{lrcll}
       \rname{ESelf} & 
                       \config{\cache, \TT, \DT, \DE, \sself, \DS} & \rightarrow &
                                                                           \config{\cache, \TT, \DT, \DE, \DE(\sself), \DS} \\

       \rname{EVar} & 
                      \config{\cache, \TT, \DT, \DE, x, \DS} & \rightarrow &
                                                                     \config{\cache, \TT, \DT, \DE, \DE(x), \DS} \\

      \rname{EAssn} &
                      \config{\cache, \TT, \DT, \DE, x=v, \DS} & \rightarrow &
                                                                               \config{\cache, \TT, \DT, \DE[x \mapsto v], v, \DS} \\

      \rname{ENew} & 
                     \config{\cache, \TT, \DT, \DE, \snew{A}, \DS} & \rightarrow &
                                                                           \config{\cache, \TT, \DT, \DE, [A], \DS} \\

      \rname{ESeq} &
                     \config{\cache, \TT, \DT, \DE, (v; e_2), \DS} & \rightarrow &
                                                                                 \config{\cache, \TT, \DT, \DE, e_2, \DS} \\

       \rname{EIfTrue} &
                        \config{\cache, \TT, \DT, \DE, \sif{v}{e_1}{e_2}, \DS} & \rightarrow &
                                                                                             \config{\cache, \TT, \DT, \DE, e_1, \DS} & \text{if}~v \neq \snil \\

       \rname{EIfFalse} &
                          \config{\cache, \TT, \DT, \DE, \sif{\snil}{e_1}{e_2}, \DS} & \rightarrow &
                                                                                                   \config{\cache, \TT, \DT, \DE, e_2, \DS} \\

      \rname{EDef} &
                     \config{\cache, \TT, \DT, \DE, \sdef{A.m}{\sprem{x}{e}}, \DS} & \rightarrow &
                                                                                         \config{\invcache{\cache}{A.m}, \TT, \DT[A.m \mapsto \lambda x.e], \DE, \snil, \DS} \\

      \rname{EType} &
                         \config{\cache, \TT, \DT, \DE, \smethsig{A.m}{\tau_m}, \DS} & \rightarrow &
                                                                                             \multicolumn{2}{l}{\config{\upgcache{(\invcache{\cache}{A.m})}{\TT'}, \TT', \DT, \DE, \snil, \DS}} \\
      \multicolumn{5}{r}{\TT' = \TT[A.m\mapsto \tau_m] \text{ and } A.m\not\in \textrm{TApp}(S)} \\
                                                                                             
       \rname{EAppMiss} &
                      \config{\cache, \TT, \DT, \DE, C[v_1.m(v_2)], \DS} & \rightarrow &
                                                                                    \multicolumn{2}{l}{\config{\cache', \TT, \DT, [\sself \mapsto v_1, x \mapsto v_2], e, (\DE, C)::\DS}} \\
      \multicolumn{5}{r}{\text{if}~A.m \not\in \dom(\cache) \text{ and } v_1 = [A] \text{ and } \DT(A.m) = \lambda x.e 
      \text{ and } \TT(A.m) = \tau_1 \rightarrow \tau_2 \text{ and } \stype{v_2} \leq \tau_1 \text{ and}} \\
      \multicolumn{5}{r}{\derivm = \left(\TT \vdash \config{[x \mapsto \tau_1, \sself \mapsto A], e}
      \Rightarrow \config{\TE', \tau}\right) \text{ holds and } \derivt = \left(\tau \leq \tau_2\right) \text{ holds and}}\\
      \multicolumn{5}{r}{\cache' = \cache[A.m\mapsto(\derivm,\derivt)]} \\

       \rname{EAppHit} &
                      \config{\cache, \TT, \DT, \DE, C[v_1.m(v_2)], \DS} & \rightarrow &
                                                                                    \multicolumn{2}{l}{\config{\cache, \TT, \DT, [\sself \mapsto v_1, x \mapsto v_2], e, (\DE, C)::\DS}} \\
      \multicolumn{5}{r}{\text{if}~A.m \in\dom(\cache) \text{ and } v_1 = [A] \text{ and } \DT(A.m) = \lambda x.e 
      \text{ and } \TT(A.m) = \tau_1 \rightarrow \tau_2 \text{ and } \stype{v_2} \leq \tau_1} \\

       \rname{ERet} &
                          \config{\cache, \TT, \DT, \DE', v, (\DE, C)::\DS} & \rightarrow &
                                                                               \config{\cache, \TT, \DT, \DE, C[v], \DS} \\




    \end{array}
  \end{displaymath}

  \begin{displaymath}
    \infer[\rname{EContext}]{
      \config{\cache, \TT, \DT, \DE, e, \DS} \rightarrow \config{\cache', \TT', \DT', \DE', e', \DS'} \\
      \nexists v_1,v_2,e' \;.\; e = (v_1.m(v_2)) \vee e = v_1 \vee e = C[e']
    }
    {\config{\cache, \TT, \DT, \DE, C[e], \DS} \rightarrow \config{\cache', \TT', \DT', \DE', C[e'], \DS'}}
  \end{displaymath}
  
  \nocaptionrule{}
  \caption{Dynamic Semantics.}
  \label{fig:semantics}
\end{figure*}

While we could potentially solve this problem with a more complex type
system, in our experience (Section~\ref{sec:experiments}) we have not
needed such a feature.

Next, \rname{TApp} types method invocation $e_0.m(e_1)$, where we look
up the method's type in $\TT$ based on the compile-time type of
$e_0$. (Note that since there is no inheritance, we need not search
the inheritance hierarchy to find the type of $A.m$.) Here subtyping
is defined as $\snil \leq A$ and $A \leq A$ for all $A$. Thus, as is
standard in languages with $\snil$, the type system may accept a
program that invokes a non-existent method of $\snil$ even though this
is a run-time error. However, notice that if $e_0$ evaluates to a
non-$\snil$ value, then \rname{TApp} guarantees $e_0$ has
method $m$.

Finally, \rname{TIf} types conditionals. Like Ruby, the guard $e_0$
may have any type. The type of the conditional is the least upper
bound of the types of the two branches, defined as $A \sqcup A = A$
and $\snil \sqcup \tau = \tau \sqcup \snil = \tau$. The output
environment of the conditional is the least upper bound of the output
environments of the branches, defined as
$(\TE_1 \sqcup \TE_2)(x) = \TE_1(x) \sqcup \TE_2(x)$ if
$x\in\dom(\TE_1) \wedge x\in\dom(\TE_2)$ and $(\TE_1 \sqcup \TE_2)(x)$
is undefined otherwise.

\paragraph*{Dynamic Semantics.}

Figure~\ref{fig:semantics} gives a small-step dynamic semantics for
our language. The semantics operates on \emph{dynamic configurations}
of the form $\config{\cache, \TT, \DT, \DE, e, \DS}$. The first two
components are the key novelties to support run-time static type
checking.  $\cache$ is a \emph{cache} mapping
$A.m$ to the type checking proofs for its method body (more
details below). $\TT$ is the \emph{type table}, which is updated at run
time by
calls to \ruby{type}. The last four components are standard. $\DT$ is
a \emph{dynamic class table} mapping $A.m$ to its premethod.
$\DE$ is the \emph{dynamic environment} mapping
local variables to values. $e$ is the expression being
reduced. Lastly, $\DS$
is a \emph{stack} of pairs $(\DE, C)$, where $\DE$ is the dynamic
environment and $C$ is the \emph{evaluation context} (defined in the
usual way) at a call site. The semantics pushes onto the stack at
calls and pops off the stack at returns.

The first seven rules in the semantics are standard. \rname{ESelf} and
\rname{EVar} evaluate $\sself$ and variables by looking them up in the
environment. \rname{EAssn} binds a variable to a value in the
environment. Notice that, like Ruby, variables can be written
without first declaring them, but it is an error to try to read a
variable that has not been written. \rname{ENew} creates a new
instance. Note that since objects do not have fields, we do not need a
separate heap. \rname{ESeq} discards the left-hand side of a sequence
if it has been fully evaluated. \rname{EIfTrue} reduces to the true
branch if the guard is non-$\snil$, and \rname{EIfFalse} reduces to
the false branch otherwise.
 
The next four rules are the heart of just-in-time static type checking.
Our goal is to
statically type check methods once at the first call, and then avoid
rechecking them unless something has changed. To formalize this
notion, we define the cache $\cache$ as a map from $A.m$ to a pair of typing
derivations $(\derivm, \derivt)$. Here $\derivm$ is a type checking
derivation from Figure~\ref{fig:types} for the body of $A.m$, and
$\derivt$ is a subtyping judgment showing that the type of $e$ is a
subtype of the declared return type. We need $\derivt$ because our
type system is syntax-directed and hence does not include a standalone
subsumption rule.

\rname{EDef} reduces to $\snil$, updating the dynamic class table to
bind $A.m$ to the given premethod along the way. Recall that we allow
a method to be redefined with \ruby{def}. Hence we need to
\emph{invalidate} anything in the cache relating to $A.m$ so that $A.m$
will be checked the next time it is called. More precisely:

\begin{definition}[Cache invalidation]
  We write $\invcache{\cache}{A.m}$ to indicate a new cache that is
  the same as $\cache$, except $A.m$ has been invalidated, meaning:
  \begin{enumerate}
    \item Any entries with $A.m$ as the key are removed.
    \item Any entries with a $\derivm$ that apply \rname{TApp} with
      $A.m$ are removed.
  \end{enumerate}
\end{definition}
Thus, in \rname{EDef}, the output cache is the same as the input cache
but with $A.m$ invalidated.

\rname{EType} also reduces to $\snil$, updating the type table to be
$\TT'$, which is the same as $\TT$ but with new type information for
$A.m$. As with \rname{EDef}, we invalidate $A.m$ in the
cache. However, there is a another subtlety. Recall that cached typing
derivations $\derivm$ include the type table $\TT$. This is
potentially problematic, because we are changing the type table to
$\TT'$. However, cache invalidation removes any derivations that
refer to $A.m$. Hence, cached type derivations that use $\TT$ can
safely use $\TT'$. Formally, we define:

\begin{definition}[Cache upgrading]
  We write $\upgcache{\cache}{\TT'}$ to indicate a new cache that is
  the same as $\cache$, except the type table in every derivation is
  replaced by $\TT'$.
\end{definition}
Thus, in \rname{EType}, the output cache is upgraded to the new type
table after invalidation.


The next two rules use the type cache. Both rules evaluate a method
call in a context, written $C[v_1.m(v_2)]$; we will discuss the other
rule for contexts shortly. In both rules, the receiver $v_1$ is a
run-time object $[A]$. \rname{EAppMiss} applies when $A.m$ is not in
the cache. In this case, we look up the type of $A.m$ in $\TT$,
yielding some type $\tau_1\rightarrow\tau_2$. We type check the method
body $e$ in an environment in which formal variable $x$ is bound to
$\tau_1$ and $\sself$ is bound to $A$, yielding a derivation
$\derivm$. We check that the resulting type $\tau$ of $e$ is a subtype
of the declared type $e_2$, with subtyping derivation
$\derivt$. Finally, we check that the run-time type of $v_2$---defined
as $\stype{\snil} = \snil$ and $\stype{[A]} = A$---is a subtype of
$\tau_1$. If all this holds, then it is type-safe to call the
method. Hence we update the cache with the method's typing derivations
and start evaluating the method body, pushing the context $C$ and the
environment $\DE$ on the stack.

\rname{EAppHit} is similar but far simpler. This rule applies when
$A.m$ is in the cache. In this case we know its method body has been
successfully type checked, so we need only check that the run-time
type of $v_2$ is a subtype of the declared domain type of $v_1$. If so,
we allow the method call to proceed.

However a method is called, the return, handled by \rname{ERet}, is
the same. This rule applies when an expression has been fully
evaluated and is at the top level. In this case, we pop the stack,
replacing $\DE'$ with $\DE$ from the stack and plugging the value $v$
into the context $C$ from the stack.

Finally, \rname{EContext} takes a step in an subexpression inside a
context $C$. This rule only applies if the subexpression is not a
method call (since that case is handled by \rname{EApp*}, which must
push the context on the stack) and not a fully evaluated value (which
is handled by \rname{ERet}, which must pop the context from the
stack). We also do not allow the subexpression to itself be a context,
since that could cause \rname{EApp*} and \rname{ERet} to misbehave.



\paragraph*{Soundness.}

Our type system forbids invoking non-existent methods of
objects. However, there are three kinds of errors the type system does
not prevent: invoking a method on $\snil$; calling a method whose body
does not type check at run time; and calling a method that has a type
signature but is itself undefined. (We could prevent the latter error
by adding a side condition to \rname{TApp} that requires the method to
be defined, but we opt not to to keep the formalism slightly simpler.)
To state a soundness theorem, we need to account for these cases,
which we do by extending the dynamic semantics with rules that reduce
to $\blame$ in these three cases. After doing so, we can state soundness:

\begin{theorem}[Soundness]
  If
  $\emptyset \vdash \config{\emptyset, e} \Rightarrow \config{\TE', \tau}$
  then either $e$ reduces to a value, $e$ reduces to $\blame$, or $e$
  diverges.
\end{theorem}
We show soundness using a standard progress and preservation
approach. The key technical challenge is preservation, in which we
need to show that not only are expression types preserved, but also
the validity of the cache and types of contexts pushed on the stack.
The proof can be found in Appendix~\ref{sec:appendix}.

\section{Implementation}
\label{sec:implementation}

\name{} is implemented using a combination of Ruby and OCaml. On the
OCaml side, we use the Ruby Intermediate Language (RIL)
\cite{furr:dls09-ril} to parse input Ruby files and translate them to
control-flow graphs (CFG) on which we perform type checking. On the
Ruby side, we extend
RDL~\cite{RDL}, a contract system
for Ruby, to perform static type checking. We next discuss the major
challenges of implementing \name{}.

\paragraph*{RIL.} RIL is essentially the front-end of Diamondback Ruby
(DRuby)~\cite{furr:oops09,DRuby}. Given
an input Ruby program, RIL produces a CFG that simplifies away many of
the tedious features of Ruby, e.g., multiple forms of conditionals. We
modified DRuby so it emits the RIL CFG as a JSON file and then
exits. When loading each application file at run-time, we read the 
corresponding JSON file and store a mapping from class and method
names and positions (file and line number) to the
JSON CFG. At run-time we look up CFGs in this map to perform
static type checking.

\paragraph*{RDL and Type Checking.} Like standard RDL, \name{}'s
\ruby{type} annotation stores type information in a map and wraps the
associated method to intercept calls to it.  We should emphasize that
RDL does not perform any static checking on its own---rather, it solely
enforces contracts dynamically.
In \name{}, when a wrapped method
is called, \name{} first checks to see if it has already been type
checked. If not, \name{} retrieves the method's CFG and type and then
statically checks that the CFG matches the given type.

\name{} uses RDL's type language, which includes nominal types,
intersection types, union types, optional and variable length
arguments, block (higher-order method) types, singleton types,
structural types, a \ruby{self} type, generics, and types for
heterogenous arrays and hashes. \name{} supports all of these kinds of
types except structural types, \ruby{self} types, heterogeneous collections,
and some variable length arguments.
In addition, \name{} adds support for both
instance field types (as seen in Figure~\ref{fig:struct}) and class
field types.

There is one slight subtlety in handling union types: If in a method
call the receiver has a union type, \name{} performs type checking
once for each arm of the union and the unions the possible return
types. For example if in call $e.m(\ldots)$ the receiver has type $A
\cup B$, then \name{} checks the call assuming $e$ has type of $A.m$,
yielding a return type $\tau_A$; checks the call assuming $B.m$,
yielding return type $\tau_B$; and then sets the call's return type to
$\tau_A \cup \tau_B$.

\paragraph{Eliminating Dynamic Checks.}
Recall the \rname{EApp*} rules dynamically check that a method's
actual arguments have the expected types before calling a statically
typed method. This check ensures that an untrusted caller cannot
violate the assumptions of a method. However, observe that if the immediate
caller is itself statically checked, then we know the arguments are
type-safe. Thus, as a performance optimization, \name{}
only dynamically checks arguments of statically typed methods if the
caller is itself not statically checked. As a further optimization,
\name{} also does not dynamically check calls from Ruby standard
library methods or the Rails framework, which are assumed to be type-safe.
The one exception is that \name{} does dynamically check types for the
Rails \ruby{params} hash, since those values come from the user's
browser and hence are untrusted.


\paragraph*{Numeric Hierarchy.}

Ruby has a Numeric tower that provides several related types for
numbers. For example,
\ruby{Fixnum} $<$ \ruby{Integer} $<$ \ruby{Numeric} and
\ruby{Bignum} $<$ \ruby{Integer} $<$ \ruby{Numeric}.
Adding two \ruby{Fixnum}s normally results in another \ruby{Fixnum},
but adding two large \ruby{Fixnum}s could result 
in a \ruby{Bignum} in the case of numeric overflow.
To keep the type checking system simple,
\name{} omits the special overflow case
and does not take \ruby{Bignum} into consideration. (This could be
addressed by enriching the type system \cite{St-Amour:2012:TNT:2187125.2187146}.)
Numeric overflow does not occur in our experiments.

\paragraph{Code Blocks.}

As mentioned earlier, Ruby code blocks are
anonymous functions delimited by \ruby{do}$\ldots$\ruby{end}. 
\name{} allows methods that take code block arguments to be annotated
with the block's type. For example:
\begin{lstlisting}[numbers=none]
type :m, "() { (T) -> U } -> nil"
\end{lstlisting}
indicates that \ruby{m} takes no regular arguments; one code block
argument where the block takes type \ruby{T} and returns type
\ruby{U}; and \ruby{m} itself returns type \ruby{nil}.

There are two cases involving code blocks that we need to
consider. First, suppose \name{} is statically type checking a call
\ruby{m() do $\mid$x$\mid$ body end}, and \ruby{m} has the type given
just above. Then at the call, \name{} statically checks that the code
block argument matches the expected type, i.e., assuming \ruby{x} has
type \ruby{T}, then \ruby{body} must produce a value of type
\ruby{U}. Second, when statically type checking \ruby{m}
itself, \name{} should check that calls to the block are type correct. Currently
this second case is unimplemented as it does not arise in our
experiments.


Recall from above that \name{} sometimes needs to dynamically check
the arguments to a statically typed method. While this test is easy to do
for objects, it is hard to do for code blocks, which would require
higher-order contracts~\cite{ff:ho-contracts}. Currently \name{} does
not implement this higher order check, and simply assumes code block
arguments are type safe. 
Also, \name{} currently assumes the \ruby{self} inside a code block is
the same as in the enclosing method body. This assumption holds in our
experiments, but it can be violated using \ruby{instance\_eval} and
\ruby{instance\_exec}~\cite{strickland:dls12}.
In the future, we plan to address this limitation by allowing the
programmer to
annotate the \ruby{self} type of code blocks.

\paragraph*{Type Casts.}

While \name{}'s type system is quite powerful, it cannot type check
every Ruby program, and thus in some cases we need to insert type
casts. \name{} includes a method \ruby{o.rdl\_cast(t)} that casts
\ruby{o}'s type to \ruby{t}. After such a call, \name{} assumes that
\ruby{o} has type \ruby{t}. At run-time, the call
dynamically checks that \ruby{o} has the given type.

In our experience, type casts have two main uses. First, sometimes
program logic dictates that we can safely downcast an object. For
example, consider the following method from one of our experiments:
\begin{lstlisting}[numbers=none]
def self.load_cache
  f = datafile_path([``cache'', ``countries''])
  t = Marshal.load(File.binread(f))
  @@cache ||= t.rdl_cast(``Hash<String, %any>'')
end
\end{lstlisting}

\ruby{Marshal.load} returns the result of converting its serialized
data argument into a Ruby object of arbitray type. However, in our
example, the argument passed to \ruby{Marshal.load} is always an application
data file that will be converted to the annotated \ruby{Hash}.

Second, by default \name{} gives instances of generic classes their
``raw'' type with no type parameters. To add parameters, we use type
casts, as in the following code:
\begin{lstlisting}[numbers=none]
a = []                      # a has type Array
a.rdl_cast("Array<Fixnum>") # cast to Array<Fixnum>
a.push(0)                   # ok
a.push("str")               # type error due to cast
\end{lstlisting}
Here without the type annotation the last line would succeed; with the
annotation it triggers a type error. Note that when casting an array or
  hash to a generic type, \ruby{rdl\_cast} iterates through the
  elements to ensure they have the given type.





\paragraph*{Modules.}

Ruby supports mixins via \emph{modules}, which are collections of
methods that can be added to a class.
%
Recall that \name{} caches which methods have been statically type
checked. Because a
module can be mixed in to multiple different classes---and can
actually have different types in those different classes---we need to
be careful that module method type checks are cached via where they
are mixed in rather than via the module name.

For example, consider the following code, where the method \ruby{foo}
defined in module \ruby{M} calls \ruby{bar}, which may vary depending
on where \ruby{M} is mixed in:
\begin{lstlisting}
module M def foo(x) bar(x) end end
class C; include M; def bar(x) x + 1 end end
class D; include M; def bar(x) x.to_s end end
\end{lstlisting}
Here method \ruby{foo} returns \ruby{Fixnum} when mixed into \ruby{C}
and \ruby{String} when mixed into \ruby{D}. Thus, rather that track
the type checking status of \ruby{M\#foo}, \name{} separately tracks
the statuses of \ruby{C\#foo} and \ruby{D\#foo}.







\paragraph*{Cache Invalidation.}

Recall from Section~\ref{sec:formalism} that \name{} needs to
invalidate portions of the cache when certain typing assumptions
change. While \name{} currently does not support cache invalidation in
general, it does support one important case. In Rails
\emph{development mode}, Rails automatically reloads modified files
without restarting, thus redefining the methods in those files but
leaving other methods intact~\cite{railsguides}. In Rails development
mode, \name{} intercepts the Rails reloading process and performs
appropriate cache invalidation. 
More specifically, when a method is called, if there is a difference
between its new and old method body (which we check using the RIL
CFGs), we invalidate the method and any methods that depend on it.
We also maintain a list of methods defined in each class, and 
when a class is reloaded we invalidate
dependencies of any method that has been removed.
In the next section, we report on an
experiment running a Rails app under \name{} as it is updated.

We plan to add more general support for cache invalidation in future
work. There are two main cases to consider. The first is when a method
is redefined or is removed (which never happens in our experiments
except in Rails development mode). Ruby provides two methods,
\ruby{method\_added} and \ruby{method\_removed}, that can be used to
register callbacks when the corresponding actions occur, which could
be used for cache invalidation.

The second case of cache invalidation is method's type changes. However,
in RDL and \name{}, multiple calls to \ruby{type} for the same method
are used to create intersection types. For example, the core
library \ruby{Array\#[]} method is given its type with the following code:
\begin{lstlisting}
  type Array, :[], '(Fixnum or Float) -> t'
  type Array, :[], '(Fixnum, Fixnum) -> Array<t>'
  type Array, :[], '(Range<Fixnum>) -> Array<t>' 
\end{lstlisting}
meaning
if given a \ruby{Fixnum} or \ruby{Float}, method \ruby{Array\#[]} returns the array
contents type; and, if given a pair of \ruby{Fixnum}s or a
\ruby{Range$<$Fixnum$>$}, it returns an array.

In this setting, we cannot easily
distinguish adding a new arm of an intersection type from replacing a
method type. Moreover, adding a new arm to an intersection type should
not invalidate the cache, since the other arms are still in
effect. Thus, full support of cache invalidation will likely require
an explicit mechanism for replacing earlier type definitions.

\section{Experiments}
\label{sec:experiments}


\begin{table*}[t!]
\begin{center}
  \begin{tabular}{|l|r||r|r|r|r|r|r|r||r|r|r|r|}
    \hline
    & & \multicolumn{3}{c|}{\textbf{Static types}} & \multicolumn{2}{c|}{\textbf{Dynamic types}} & & & \multicolumn{4}{c|}{\textbf{Running time (s)}} \\ \cline{3-7}\cline{10-13}
    \textbf{App} & \textbf{LoC} & \textbf{Chk'd} & \textbf{App} & \textbf{All} & \textbf{Gen'd} & \textbf{Used} & \textbf{Casts} & \textbf{Phs}& \textbf{Orig} & \textbf{No\$} & \textbf{Hum} & \textbf{Or. Ratio} \\ \hline
    \sname{Talks-1/4/2013} & 1,055 & 111 & 201 & 363 & 990 & 45 & 31 & 1 & 162 & 1,590 & 256 & 1.6$\times$ \\
    \sname{Boxroom-1.7.1} & 854 & 127 & 221 & 306 & 534 & 93 & 17 & 1 & 263 & 705 & 327 & 1.2$\times$ \\
    \sname{Pubs-1/12/2015} & 620 & 47 & 86 & 171 & 445 & 33 & 13 & 1 & 72.0 & 4,470 & 217 & 3.0$\times$ \\
    \sname{Rolify-4.0.0}  & 84 & 14 & 24 & 71 & 26 & 2 & 15 & 12 & 5.63 & 7.79 & 6.71 & 1.2$\times$ \\
    \sname{CCT-3/23/2014} & 172 & 23 & 27 & 75 & 6 & 3 & 6 & 1 & 3.06 & 78.2 & 17.4 & 5.7$\times$ \\
    \sname{Countries-1.1.0} & 227 & 33 & 40 & 111 & 0 & 0 & 22 & 1 & 1.02 & 18.1 & 4.62 & 4.5$\times$ \\
    \hline
  \end{tabular}
 \end{center}
\nocaptionrule{}
 \caption{Type checking results.}
 \label{table:type-check-results}
\end{table*}

We evaluated \name{} by applying it to six Ruby apps:
\begin{itemize}
\item \sname{Talks}\footnote{\url{https://github.com/jeffrey-s-foster/talks}}
  is a Rails app, written by the second author, for publicizing talk
  announcements. Talks has been in use in the UMD CS department since
  February 2012.
\item \sname{Boxroom}\footnote{\url{http://boxroomapp.com}} is a Rails
  implementation of a simple file sharing interface.
\item \sname{Pubs} is a Rails app, developed several years ago by the
  second author, for managing lists of publications.
\item
  \sname{Rolify}\footnote{\url{https://github.com/RolifyCommunity/rolify}}
  is a role management library for Rails. For this evaluation, we
  integrated \sname{Rolify} with \sname{Talks} on the \ruby{User} resource.
\item Credit Card Transactions (\sname{CCT})\footnote{\url{https://github.com/daino3/credit_card_transactions}} is a library
  that performs simple credit card processing tasks.
\item
  \sname{Countries}\footnote{\url{https://github.com/hexorx/countries}}
  is an app that provides useful data about each country.
\end{itemize}
We selected these apps for variety rather than for being
representative. We chose these apps because their source code is
publicly available (except \sname{Pubs}); they work with the latest
versions of Ruby and RDL; and they do not rely heavily on other
packages. Moreover, the first three apps use Rails, which is an
industrial strength web app framework that is widely deployed;
the next two use various metaprogramming styles
in different ways than Rails; and the last one does not use
metaprogramming, as a baseline.



Table~\ref{table:type-check-results} summarizes the results of
applying \name{} to these apps. On the left we list the app name,
version number or date (if no version number is available), and lines
of code as measured with \ruby{sloccount}~\cite{sloccount}.
For the Rails apps, we ran \ruby{sloccount} on
all ruby files in the model, controller, helper, and mailer directories.
We do not include lines of code for views, as we do not type check views.
For \sname{Countries} and \sname{CCT}, we ran \ruby{sloccount}
on all files in the \ruby{lib} directory.
For \sname{Rolify}, we only statically checked several
methods, comprising 84 lines of code, that use \ruby{define\_method}
in an interesting way.


\paragraph*{Type Annotations.}

For all apps, we used common type annotations from RDL for the Ruby
core and standard libraries.
For several apps, we also added type annotations for
third-party libraries and for the Rails framework. We trusted the
annotations for all these libraries, i.e., we did not statically type
check the library methods' bodies.

We also added code to dynamically generate types for metaprogramming
code. For Rails, we added code to dynamically generate types for model
getters and setters based on the database schema; for finder methods
such as \ruby{find\_\-by\_\-name} and \ruby{find\_all\_by\_password} (the
method name indicates which field is being searched); and for Rails
associations such as \ruby{belongs\_to}.


In Figure~\ref{fig:define_method}, we showed code we added to
\sname{Rolify} to generate types for a method created by calling
\ruby{de\-fine\_dy\-nam\-ic\_meth\-od}. Calling \ruby{de\-fine\_\-dy\-nam\-ic\_\-meth\-od}
also dynamically creates another method,
\ruby{is\_\#\{role\_\-name\}\_of(arg)?}, which we also provide types
for in the \ruby{pre} block.



In \sname{CCT}, we used the code in Figure~\ref{fig:struct} to generate types
for \ruby{Struct} getters and setters.



Finally, we wrote type annotations for the app's own methods that were
included in the lines of code count in
Table~\ref{table:type-check-results}. We marked those methods to
indicate \name{} should statically type check their bodies.
Developing these annotations was fairly straightforward, especially
since we could quickly detect incorrect annotations by running
\name{}.

\paragraph*{Type Checking Results.}

For each program, we performed type checking while running unit tests
that exercised all the type-annotated app methods.  For \sname{Talks}
and \sname{Pubs}, we wrote unit tests with the goal of covering all
application methods. For
\sname{Boxroom}, we used its unit tests on models but wrote our own
unit tests on controllers, since it did not have controller tests. For \sname{Rolify}, we wrote
a small set of unit tests for the dynamic method definition feature. 
For \sname{CCT} and \sname{Countries},
we used the unit tests that came with those apps.

In all cases, the app methods type check correctly in \name{}; there
were no type errors. The middle group of columns summarizes more
detailed type checking data.

The ``Static types'' columns report data on static type
annotations. The count under ``Chk'd'' is the number of type
annotations for the app's methods whose bodies we statically type checked. 
The count under ``App'' is that number plus the number of types for
app-specific methods with (trusted) static type annotations, e.g.,
some Rails helper functions have types that we do not currently
dynamically generate.
The count under ``All'' reports the total number of
static type annotations we used in type checking each app. This
includes the ``App'' count plus standard, core, and third-party
library type annotations for methods referred to in the app.

The ``Dynamic types'' columns report the number of types that were
dynamically generated (``Gen'd'') and the number of those that were
actually used during type checking (``Used''). These numbers differ
because we tried to make the dynamic type information general rather
than app-specific, e.g., we generate both the getter and setter for
\ruby{belongs\_to} even if only one is used by the app.

These results show that having types for methods generated by
metaprogramming is critical for successfully typing these
programs---every app except \sname{Countries} requires at least a few,
and sometimes many, such types.

The ``Casts'' column reports the number of type casts we
needed to make these programs type check; essentially this measures
how often \name{}'s type system is overly conservative. The results
show we needed a non-trivial but relatively small number of
casts. All casts were for the reasons discussed in
Section~\ref{sec:implementation}: downcasting and
generics. 



The ``Phs'' column in Table~\ref{table:type-check-results} shows
the number of type checking phases under \name{}. Here a \emph{phase}
is defined as a sequence of type annotation calls with no intervening
static type checks, followed by a sequence of static type checks with
no intervening annotations. We can see that almost all apps have only
a single phase, where the type annotations are executed before any
static type checks. Investigating further, we found this is due to the
way we added annotations. For example, we set up our Rails apps so the
first loaded application file in turn loads all type annotation
files. In practice the type annotations would likely be spread
throughout the app's files, thus increasing the number of
phases.


\sname{Rolify} is the only application with multiple phases. Most of
the phases come from calling \ruby{define\_dynamic\_method}, which
dynamically defines other methods and adds their type annotations.
The other phases come from the order in which the type annotation
files are required---unlike the Rails apps, the \sname{Rolify} type annotation files
are loaded piecemeal as the application loads.


\paragraph*{Performance.}
The last four columns of Table~\ref{table:type-check-results} report the
overhead of using \name{}. The ``Orig'' column shows the running time
without \name{}. The next two columns report the running time with
\name{}, with caching disabled (``No\$'') and enabled
(``Hum''). The last column lists the ratio of the ``Hum'' and
``Orig'' column.

For \sname{Talks}, \sname{Boxroom}, and \sname{Pubs}, we measured the
running time of a client script that uses \ruby{curl} to connect to
the web server and exercise a wide range of functionality.  For
\sname{CCT}, we measured the time for running its unit tests 100
times.  For \sname{Countries} and \sname{Rolify}, we measured the time
for running the unit tests once (since these take take much more time
than \sname{CCT}'s tests). For all apps, we performed each measurement
three times and took the arithmetic mean.

These results show that for the Rails apps, where IO is significant,
\name{} slows down performance from 24\% to 201\% (with
caching enabled).
We think these are reasonable results for an early
prototype that we have not spent much effort optimizing. Moreover,
across all apps, the ratios
are significantly better than prior systems that mix static and
dynamic typing for Ruby~\cite{an:popl11,ren:oops13}, which report
orders of magnitude slowdowns.

Investigating further, we found that the main \name{} overhead arises
from intercepting method calls to statically type checked
methods. (Note the interception happens regardless of the cache
state.) The higher slowdowns for \sname{CCT} and \sname{Countries}
occur because those applications spend much of their time in
code with intercepted calls, while the other applications spend most
of their time in framework code, whose methods are not intercepted.
We expect performance can be improved with further engineering effort.

We can also see from the results that caching is an important
performance optimization: without caching, performance slows down
1.4$\times$ to 62$\times$. We investigated \sname{pubs}, the app with
the highest no-caching slowdown, and found that while running the application
with large array inputs, certain application methods are called more
than 13,000 times while iterating through the large arrays.
This means that each of these application methods are statically
type checked more than 13,000 times when caching is disabled.



\paragraph*{Type Errors in \sname{Talks}.} We downloaded many earlier
versions of \sname{Talks} from its \ruby{github} repository and ran
\name{} on them using mostly the same type annotations as for the latest
version, changed as necessary due to program changes.  Cumulatively,
we found six type errors that were introduced and later removed as
\sname{Talks} evolved. Below the number after the date indicates which
checkin it was, with 1 for the first checkin of the day, 2 for the
second, etc.
\begin{itemize}
\item 1/8/12-4: This version misspells \ruby{compute\_edit\_fields}
  as \ruby{copute\_edit\_fields}. \name{} reported this error because
  the latter was an unbound local variable and was also not a valid method.
\item 1/7/12-5: Instead of calling
  \ruby{@list.talks.upcoming.sort\{$|$ a, b $|$ ...\}}, this version calls
  \ruby{@list.talks.upcoming\{$|$ a, b $|$ ...\}} (leaving off the
  \ruby{sort}). \name{} detects this error because \ruby{upcoming}'s
  type indicates it does not take a block. Interestingly, this error
  would not be detected at run-time by Ruby, which simply ignores
  unused block arguments.
\item 1/26/12-3: This version calls \ruby{user.subscribed\_talks(true)},
  but \ruby{subscribed\_talks}'s argument is a \ruby{Symbol}.
\item 1/28/12: This version calls \ruby{@job.handler.object}, but 
  \ruby{@job.handler} returns a \ruby{String}, which does not have an
   \ruby{object} method.
\item 2/6/12-2: This version uses undefined variable \ruby{old\_talk}.
Thus, \name{} assumes \ruby{old\_talk} is a no-argument method and attempts to
look up its type, which does not exist.
\item 2/6/12-3: This version uses undefined variable \ruby{new\_talk}
\end{itemize}

We should emphasize that although we expected there would be type
errors in \sname{Talks}, we did not know exactly what they were or
what versions they were in. While the second author did write
\sname{Talks}, the errors were made a long time ago, and the second
author rediscovered them independently by running \name{}.

\paragraph*{Updates to Talks}
Finally, we performed an experiment in which we launched one version
of \sname{Talks} in Rails development mode and then updated the code
to the next six consecutive versions of the app.  (We skipped
versions in which none of the Ruby application files changed) Notice
that cache invalidation is particular useful here, since in typical
usage only a small number of methods are changed by each update.


In more detail, after launching the initial version of the app, we
repeated the following sequence six times: Reset the database (so that
we run all versions with the same initial data);
run a sequence of curl commands that access the same \sname{Talks}
functionalities as the ones used to measure the running time of
\sname{Talks} in Table~\ref{table:type-check-results}
; update the code to the next version; and repeat.



Table~\ref{table:live-update-exp} shows the results of our experiment.
The ``$\Delta$
Meth'' column lists the number of methods whose bodies or types were
changed compared to the previous version. 
Note there are no removed methods in any of these versions.
The ``Added'' column lists the number of methods added; such methods
will be checked when they are called for the first time but 
do not cause any cache invalidations.
The ``Deps'' column counts the number of \emph{dependent} methods that call
one or more of the changed methods. These methods plus the changed
methods are those whose previous static type check are invalidated by
the update. The last column, ``Chk'd,''
reports how many methods are newly or re-type checked after the
update. Currently, \name{} always rechecks Rails helper methods, due
to a quirk in the Rails implementation---the helper methods' classes
get a new name each time the helper file is reloaded, causing \name{}
to treat their methods as new. Thus (except for the first line, since
this issue does not arise on the first run), we list two numbers in the
column: the first with all rechecks, including the helper methods, and
the second excluding the helper methods.

These results show that in almost all cases, the second number in
``Chk'd'' is equal to the sum of the three previous columns. There is
one exception: in 8/24/12/-1, there 14 rechecked methods but 18
changed/added/dependent methods. We investigated and found that the 14
rechecks are composed of six changed methods that are rechecked once;
two changed methods that are rechecked twice because they have
dependencies whose updates are interleaved with calls to those
methods; one added method that is checked; and three dependent methods
that are rechecked. The remaining added method is not called by the
curl scripts, and the remaining dependent methods are also changed
methods (this is the only version where there is overlap between the
changed and dependent columns).

Finally,
as there are no type errors in this sequence of updates, we confirmed
that this streak of updates type checks under \name{}.
 
\begin{table}[t!]
\begin{center}
  \begin{tabular}{|l|r|r|r|r|r|r||r|r|r|r|}
    \hline
     \textbf{Version} & \textbf{$\Delta$ Meth} & \textbf{Added} & \textbf{Deps} & \textbf{Chk'd}\\ \hline
     5/14/12   & N/A & N/A & N/A & 77 \\
     7/24/12   & 1 & - & 4 & 15\;/\;5 \\
     8/24/12-1 & 8 & 2 & 8 & 24\;/\;14 \\
     8/24/12-2 & - & 1 & - & 11\;/\;1 \\
     8/24/12-3 & 1 & 1 & - & 12\;/\;2 \\
     9/14/12   & 1 & - & - & 15\;/\;1 \\
     1/4/13 & 4 & - & - & 13\;/\;4 \\
    \hline
  \end{tabular}
 \end{center}
\nocaptionrule{}
 \caption{Talks Update Results}
 \label{table:live-update-exp}
\end{table}

\section{Related Work}
\label{sec:related}

There are several threads of related work.

\paragraph*{Type Systems for Ruby.} We have developed several prior type systems
for Ruby. Diamondback Ruby (DRuby)~\cite{furr:oops09} is the first
comprehensive type system for Ruby that we are aware
of.
Because \name{} checks types at run-time, we opted to implement our
own type checker rather than reuse DRuby for type checking, which would have
required some awkward shuffling of the type table between Ruby and
OCaml. Another reason to reimplement type checking was to keep the
type system a little easier to understand---DRuby performs type
inference, which is quite complex for this type language, in contrast
to \name{}, which implements much simpler type checking.

DRuby was effective but did not handle highly dynamic language
constructs well. \pruby{}~\cite{furr:oopsla09} solves this problem using profile-based type
inference. To use \pruby{}, the developer runs the
program once to record dynamic behavior, e.g., what methods are
invoked via \ruby{send}, what strings are passed to \ruby{eval}, etc.
\pruby{} then applies DRuby to the original program text plus the
profiled strings, e.g., any string that was passed to \ruby{eval} is
parsed and analyzed like any other code. While \pruby{} can be effective,
we think that \name{}'s approach is ultimately more practical because
\name{} does not require a separate, potentially cumbersome, profiling
phase. We note that \name{} does not currently handle \ruby{eval},
because it was not used in our subject apps' code, but it could be
supported in a straightforward way.

We also developed DRails~\cite{an:ase09}, which type checks Rails apps
by applying DRuby to translated Rails code.  For example, if DRails
sees a call to \ruby{belongs\_to}, it outputs Ruby code that
explicitly contains the methods generated from the call, which DRuby
can then analyze. While DRails was applied to a range of programs, its
analysis is quite brittle. Supporting each additional Rails feature in
DRails requires implementing, in OCaml, a source-to-source
transformation that mimics that feature.  This is a huge effort and is
hard to sustain as Rails evolves. In contrast, \name{} types are
generated in Ruby, which is far easier.  DRails is also complex to
use: The program is combined into one file, then run to gather profile
information, then transformed and type checked. Using \name{} is far
simpler. Finally, DRails is Rails-specific, whereas \name{} applies
readily to other Ruby frameworks. Due to all these issues, we feel \name{}
is much more lightweight, agile, scalable, and maintainable than
DRails.

Finally, RubyDust~\cite{an:popl11} implements type inference for Ruby
at run time. RubyDust works by wrapping objects to annotate them with
type variables. More precisely, consider a method \ruby{def m(x)
  ... end}, and let $\alpha$ be the type variable for~\ruby{x}.
RubyDust's wrapping is approximately equal to adding
\ruby{x = Wrap.new(x, $\alpha$)} to the beginning of \ruby{m}.
Uses of the wrapped \ruby{x} generate type constraints on $\alpha$ and
then delegate to the underlying object. The Ruby Type
Checker~\cite{ren:oops13} (rtc) is similar but implements type
checking rather than type inference.

\name{} has several important advantages over RubyDust and rtc. First,
RubyDust and rtc can only report errors on program paths they
observe. In contrast, \name{} type checks all paths through methods it
analyzes. Second, wrapping every object with a type annotation is
extremely expensive. By doing static analysis, \name{} avoids this
overhead. 
Finally, RubyDust and rtc have no special support for
metaprogramming. In RubyDust, dynamically created methods could have
their types inferred in a standard way, though RubyDust would likely not
infer useful types for Rails-created methods. In rtc, dynamically
created methods would lack types, so their use would not be
checked. (Note that it would certainly be possible to add \name{}-style
support for metaprogramming-generated type annotations to either
RubyDust or rtc.)
In sum, we think that \name{} strikes the right compromise between
the purely static DRuby approach and the purely dynamic RubyDust/rtc
approach.

\paragraph*{Type Systems for Other Dynamic Languages.}
Many researchers have proposed type systems for dynamic languages,
including Python~\cite{aycock:python},
JavaScript~\cite{anderson:javascript,thiemann:javascript,Lerner:2013:TRT:2508168.2508170},
Racket~\cite{Tobin-Hochstadt:2008:DIT:1328438.1328486,St-Amour:2012:TNT:2187125.2187146,Tobin-Hochstadt:2010:LTU:1863543.1863561}, and
Lua~\cite{Maidl:2014:TLO:2617548.2617553}, or developed new dynamic languages or
dialects with special type systems, such as 
Thorn~\cite{bloom-vitek:thorn},
TypeScript~\cite{bierman2014understanding,Rastogi:2015:SEG:2676726.2676971},
and Dart~\cite{dartlang}. To our knowledge, these type systems are focused
on checking the core language and can have difficulty in the face of
metaprogramming.

One exception is RPython~\cite{ancona:rpython}, which
introduces a notion of load time, during which highly dynamic features
may be used, and run time, when they may not be. In contrast, \name{}
does not need such a separation.

Lerner et al~\cite{Lerner:2013:CFF:2524984.2524990} propose a system for type checking
programs that use JQuery, a very sophisticated Javascript
framework. The proposed type system has special support for JQuery's
abstractions, making it quite effective in that domain. On the other
hand, it does not easily apply to other frameworks.

Feldthaus et al's TSCHECK~\cite{Feldthaus:2014:CCT:2660193.2660215} is
a tool to check the correctness of TypeScript interfaces for JavaScript 
libraries. TSCHECK discovers a library's API by taking a snapshot after 
executing the library's top-level code. It then performs checking using 
a separate static analysis. This is similar to \name{}'s tracking 
of type information at run-time and then performing static checking 
based on it. However, \name{} allows type information to be 
generated at any time and not just in top-level code.

\paragraph*{Related Uses of Caching.} Several researchers have
proposed systems that use caching in a way related to \name{}.
Koukoutos et al~\cite{DBLP:conf/rv/KoukoutosK14} reduce the overhead
of checking data structure contracts (e.g., ``this is a binary search
tree'') at run time by modifying nodes to hold key verification
properties. This essentially caches those properties. However, because
the properties are complex, the process of caching them is not
automated.

Stulova et al~\cite{DBLP:journals/corr/StulovaMH15} propose memoizing
run-time assertion checking to improve performance.  This is similar
to \name{}'s type check caching, but much more sophisticated because
the cached assertions arise from a rich logic.

Hermenegildo et al~\cite{Hermenegildo:2000:IAC:349214.349216} proposed
a method to incrementally update analysis results at run-time as code
is added, deleted, or changed.
Their analysis algorithms are designed for constraint logic
programming languages, and are
much more complicated than \name{}'s type checking.


\paragraph*{Staged Analysis.}

MetaOCaml~\cite{MetaOcaml} is a multi-stage extension of OCaml in which code is compiled
in one stage and executed in a later stage. The MetaOCaml compiler 
performs static type checking on any such delayed code, which is similar 
to \name{}'s just-in-time type checking. A key difference between MetaOCaml 
and \name{} is that Ruby programs do not have clearly delineated stages.

Chugh et al's staged program analysis~\cite{Chugh:2009:SIF:1542476.1542483}
performs static analysis on as much 
code as is possible at compile time, and then computes a set of remaining 
checks to be performed at run time. \name{} uses 
a related idea in which no static analysis is performed at compile time, 
but type checking is always done when methods are called. \name{} is 
simpler because it need not compute which checks are necessary, as it always 
does the same kind of checking.

\paragraph*{Other.}

Several researchers have explored other ways to bring the benefits of
static typing to dynamic languages. Contracts~\cite{ff:ho-contracts}
check assertions at function or method entry and exit. In contrast,
\name{} performs static analysis of method bodies, which can find bugs
on paths before they are run. At the same time, contracts can encode
richer properties than types.

Gradual typing~\cite{siek06:gradual} lets developers add types
gradually as programs evolve; Vitousek et al recently implemented
gradual typing for
Python~\cite{Vitousek:2014:DEG:2661088.2661101}. Like
types~\cite{Wrigstad:2010:ITU:1706299.1706343} bring some of the
flexibility of dynamic typing to statically typed languages.
The goal of these systems is to allow mixing of typed and untyped
code. This is orthogonal to \name{}, which focuses on checking code
with type annotations.

Richards et al~\cite{Richards:2010:ADB:1806596.1806598,Richards:2011:EML:2032497.2032503} have explored how highly dynamic language features are
used in
JavaScript. They
find such features, including \ruby{eval}, are used extensively in a
wide variety of ways, including supporting metaprogramming.

The GHC Haskell compiler lets developers defer type errors until run-time 
to suppress type errors on code that is never actually executed \cite{ghc}. 
\name{} provides related behavior in that a method that is never called 
will never be type checked by \name{}.  Template Haskell 
\cite{Sheard:2002:TMH:636517.636528} can be used 
for compile-time metaprogramming. Since Haskell programs contain types, 
template Haskell is often used to generate type annotations, analogously 
to the type annotations generated using \name{}. Similarly, F\# type 
providers \cite{fsharp} 
allow users to create compile time types, properties and methods.
A key difference between these Haskell/F\# features and \name{} is that
Ruby does not have a separate compile time.

\section{Conclusion}

We presented \name{}, a novel tool that type checks Ruby apps using an
approach we call just-in-time static type checking.
\name{} works by tracking
type information \emph{dynamically}, but then checking method bodies
\emph{statically} at run time as each method is called. As long as
any metaprogramming code is extended to generate types as it creates
methods, \name{} will, in a very natural way, be able to check code
that uses the generated methods. Furthermore, \name{} can cache type
checking so it need not be unnecessarily repeated at later calls to the same method.

We formalized \name{} using a core, Ruby-like language that allows
methods and their types to be defined at arbitrary (and arbitrarily
separate) points during execution, and we proved type soundness.
We implemented \name{} on top of RIL, for parsing Ruby source code,
and RDL, for intercepting method calls and storing type
information. We applied \name{} to six Ruby apps, some of which use
Rails. We found that \name{}'s approach is effective, allowing it to
successfully type check all the apps even in the presence of
metaprogramming. We ran \name{} on earlier versions of one app
and found several type errors. Furthermore, we ran \name{} while
applying a sequence of updates to a Rails app in development mode to
demonstrate cache invalidation under \name{}.
Finally, we measured \name{}'s run-time overhead and found it is reasonable.

In sum, we think that \name{} takes a strong step forward in bringing
static typing to dynamic languages.

\section*{Acknowledgments}

Thanks to ThanhVu Nguyen and the anonymous reviewers for their helpful
comments. This research was supported in part by NSF CCF-1319666 and
Subcontract to Northeastern University, NSF CCF-1518844.


\bibliographystyle{abbrvnat}
\bibliography{tr}

\appendix

\section{Complete Formalism}
\label{sec:appendix}

This section contains the full definitions and proofs for the
formalism.



We show soundness by first showing preservation and progress. As is
typical, the hardest part of the proof is preservation, which shows
that an expression's type is preserved under a step in the dynamic
semantics. To make the theorem work, we also need to reason about
preserving key properties about the typing environment, run-time
stack, and cache. Here is the statement of the theorem, which
we explain in detail next:

\begin{theorem}[Preservation]
  If
  \begin{enumerate}[label=(\arabic*)]
    \item $\config{\cache, \TT, \DT, \DE, e, \DS} \rightarrow
      \config{\cache', \TT', \DT', \DE', e', \DS'}$
    \item $\TT \vdash \config{\TE, e} \Rightarrow \config{\TE', \tau}$
    \item $\tau \leq \TS$
    \item $\TE \sim \DE$
    \item $\TT \vdash \TS \sim \DS$
    \item $\cache \sim (\TT, \DT)$
  \end{enumerate}
  Then there exist $\TEalt, \TEalt', \TS', \tau'$ such that
  \begin{enumerate}[label=(\alph*)]
  \item $\TT' \vdash \config{\TEalt, e'} \Rightarrow \config{\TEalt', \tau'}$
  \item $\tau' \leq \TS'$
  \item If $\DS = \DS'$ then $\TEalt' \leq \TE'$
  \item $\TEalt \sim \DE'$
  \item $\TT' \vdash \TS' \sim \DS'$
  \item $\cache' \sim (\TT', \DT')$
  \end{enumerate}
  
\end{theorem}

Let's step through the assumptions and conclusions of the theorem.
(1) and (2) are standard---they assume that $e$ takes a step and is
well-typed, respectively. The corresponding conclusion (a) states that
$e'$ is also well-typed.

(4) assumes the type and dynamic environments are consistent---meaning
values in $\DE$ have the corresponding types in $\TE$---and conclusion (d)
states that they are still consistent after reduction. Formally:

\begin{definition}[Environment consistency]
  Type environment $\TE$ is \emph{consistent} with dynamic environment
  $\DE$, written $\TE \sim \DE$, if $\dom(\TE) \subseteq \dom(\DE)$
  and for all $x\in\dom(\TE)$ there exists $\tau$ such that
  $\cdot \vdash \config{\TE, \DE(x)} \Rightarrow \config{\TE, \tau}$
  and $\tau \leq \TE(x)$.
\end{definition}

Notice this definition allows $\DE$ to include some variables that are
not bound in $\TE$. This is necessary to handle \rname{TIf}, which
discards any variables from the type environment that are bound in one
arm of the conditional but not the other.

Next, (3) and (5) concern the type of $e$ and the
stack. The goal of preservation is to show $e$'s type is preserved,
but consider \rname{EApp*} and \rname{ERet}. These rules both push and
pop the stack and change the expression being evaluated---hence $e'$
could potentially have an entirely different type than $e$.

Our solution is to introduce the notion of a \emph{type stack} $\TS$
to mirror the run-time stack. To understand how the type stack works,
suppose we want to apply preservation to $C[v_1.m(v_2)]$, i.e., we are
about to call a method. The typing judgment is
$\TT \vdash \config{\TE, C[v_1.m(v_2)} \Rightarrow \config{\TE',
  \tau'}$.
In the dynamic semantics, the \rname{EApp*} rules will push the
current environment $\DE$ and the context $C$ on the
stack. Correspondingly, we will push the current typing judgment onto
the type stack---at least the key pieces of it. More specifically, we
push an element of the form $\tselt{\TE}{\tau}{\TE'}{\tau'}$, where
$\TE$ and $\TE'$ are the initial and final environments of the current
typing judgment; $C$ is the context; and $\tau$ is the type of
expression $v_1.m(v_2)$, i.e., the type that the method must return.

Given this mechanism, the key invariant to maintain is that the type
of the expression is compatible with what the calling functions
expects. We define:

\begin{definition}[Stack subtyping]
  $\tau_0 \leq \tselt{\TE}{\tau}{\TE'}{\tau'}::\TS$ if $\tau_0 \leq \tau$.
\end{definition}

Then (3) assumes that the type of $e$ is a subtype of the type
expected by the calling function. (At the top-level, we initialize the
type stack with a frame that expects whatever the top-level type is.)
(b) states that the type of $e'$ is also a subtype of the type
expected by its calling function. Thus, if the stack does not change,
this means that $e'$ and $e$ have the same type (up to subtyping). If
the stack does change, then we still maintain the invariant.

Of course, we need this invariant to hold no matter how many pushes
and pops happen. Thus, rather than only talk about the top element of
the type stack, we need to ensure that all elements of the type stack
are consistent with all elements of the dynamic stack. Formally:

\begin{definition}[Stack consistency]
  Type stack element \\ $\tselt{\TE}{\tau}{\TE'}{\tau'}$ is \emph{consistent} with dynamic
  stack element $(E, C)$, written
  $\TT \vdash \tselt{\TE}{\tau}{\TE'}{\tau'} \sim (E, C)$, if $\TE \sim E$ and
  $\TT \vdash \config{\TE[\Box\mapsto\tau], C} \Rightarrow \config{\TE', \tau'}$.
  (Here we abuse notation and treat $\Box$ as if it's a variable.)
  
  Type stack $\TS$ is \emph{consistent} with dynamic stack $\DS$,
  written $\TT \vdash \TS \sim \DS$, is defined inductively as
  \begin{enumerate}
    \item $\TT \vdash \cdot \sim \cdot$
    \item $\TT \vdash \tselt{\TE}{\tau}{\TE'}{\tau'}::\TS \sim
      (E,C)::\DS$ if
      \begin{enumerate}
      \item $\tselt{\TE}{\tau}{\TE'}{\tau'} \sim (E,C)$
      \item $\TS \sim \DS$
      \item $\tau' \leq \TS$ if $\TS \neq \cdot$
      \end{enumerate}
    \end{enumerate}
\end{definition}

Thus, (5) assumes the type and dynamic stacks are consistent, and (e)
concludes they remain consistent after taking a step.

Next, (c) relates the output environment of $e'$ with the output
environment of $e$.There are two cases. If the stack did not change
(the antecedent of the conclusion is true), then the output
environment of $e'$ should be compatible with $\TE'$. Again because of
\rname{TIf}, we need to allow the environment to shrink:

\begin{definition}[Type environment subsumption]
  We write $\TE_1 \leq \TE_2$ if $\dom(\TE_2) \subseteq \dom(\TE_1)$ and
  for all $x\in\dom(\TE_2)$, it is the case that $\TE_1(x) \leq \TE_2(x)$.
\end{definition}

If the stack does change, then the output environment is irrelevant:
It either is captured in the type stack if this is a push due to a
method call. Or it is discarded as the stack frame is popped when a
method returns. Hence in this case the antecedent of (c) is false, and
the conclusion is trivial.

Finally, we need to reason about the cache. As we saw earlier, the key
cache invariant to preserve is that all the derivations stored in the
cache hold and apply to the premethod stored in $\DT$ and the type
stored in $\TT$. Formally:

\begin{definition}[Cache consistency]
  We say that cache $\cache$ is \emph{consistent} with type class
  table $\TT$ and dynamic class table $\DT$, written
  $\cache \sim (\TT, \DT)$, if for all $A.m\in\dom(\cache)$ where
  $\cache(A.m) = (\derivm, \derivt)$, 
with $\derivm = (\TT \vdash \config{[x\mapsto\tau_1, \sself\mapsto A],
  e} \Rightarrow \config{\TE', \tau})$ and $\derivt = (\tau \leq \tau_2)$,
it is the case that
  $\derivm$ and $\derivt$ hold and
  $\DT(A.m) = \lambda x.e$ and
  $\TT(A.m) = \tau_1 \rightarrow \tau_2$.
\end{definition}

Thus, (6) assumes the cache is consistent, and (f) concludes the new
cache is also consistent.

To show preservation, we also need a few lemmas:

\begin{lemma}
  \label{lemma:type-env-leq-union}
  For all $\TE_1$ and $\TE_2$, it is the case that
  $\TE_1 \leq (\TE_1 \sqcup \TE_2)$.
\end{lemma}

\begin{lemma}[Contextual Substitution]
  If
  \begin{displaymath}
    \infer
    {
      \TT \vdash \config{\TE, e} \Rightarrow \config{\TE,
        \tau'} \\\\
      \vdots
    }
    {\TT \vdash \config{\Gamma_C, C[e]} \Rightarrow \config{\Gamma'_C, \tau_C}}
  \end{displaymath}
  then $\TT \vdash \config{\Gamma_C[\Box\mapsto\tau'], C} \Rightarrow
  \config{\Gamma'_C, \tau_C}$.
\end{lemma}

\begin{lemma}[Substitution]
  If
  \begin{enumerate}
    \item $\TT \vdash \config{\TEalt[\Box\mapsto\tau_C], C}
      \Rightarrow \config{\TEalt', \tau'_C}$
    \item $\TT \vdash \config{\cdot, v} \Rightarrow \config{\cdot, \tau}$
    \item $\tau \leq \tau_C$
  \end{enumerate}
  Then
  \begin{math}
    \TT \vdash \config{\TEalt, C[v]} \Rightarrow \config{\TEalt', \tau''_C}
  \end{math}
  where $\tau''_C \leq \tau'_C$.
\end{lemma}

Finally, we can prove preservation:

\begin{proof} \textbf{(Preservation)}
  By induction on $\config{\cache, \TT, \DT, \DE, e, \DS} \rightarrow$ \\
  $\config{\cache', \TT', \DT', \DE', e', \DS'}$.

  \begin{itemize}

  \item Case \rname{EContext}. Notice that we cannot have
    $\DS' \neq \DS$, since the only cases where that can happen is if
    \rname{EApp} or \rname{ERet} apply, and they cannot be used as a
    hypothesis of \rname{EContext}. Thus the left-hand side of the
    implication (c) is true, and we have $\TEalt' \leq \TE'$.  Using
    this fact, the remainder of the proof is routine.

    \item Case \rname{ESelf}. By assumption we have
      \begin{enumerate}[label=(\arabic*)]
      \item $\config{\cache, \TT, \DT, \DE, \sself, \DS} \rightarrow
        \config{\cache, \TT, \DT, \DE, \DE(\sself), \DS}$ by \rname{ESelf}
      \item $\TT \vdash \config{\TE, \sself} \Rightarrow \config{\TE,
          \TE(\sself)}$ by \rname{TSelf}
      \item $\TE(\sself) \leq \TS$
      \item $\TE \sim \DE$
      \item $\TT \vdash \TS \sim \DS$
      \item $\cache \sim (\TT, \DT)$
      \end{enumerate}
      Let $\TEalt = \TEalt' = \TE$, and let $\TS' = \TS$.  By (2) and
      (4) there exists $\tau'$ such that
      $\cdot \vdash \config{\TEalt, \DE(\sself)} \Rightarrow
      \config{\TEalt, \tau'}$
      and $\tau' \leq \TEalt(\sself)$.  Then (a) holds, since typing
      of $\DE(\sself)$ was by \rname{TNil} or \rname{TObject}, which
      do not depend of the type class table.  Also, (b) holds since
      $\tau' \leq \TEalt(\sself) = \TE(\sself) \leq TS$ by (3).
      Also, the right-hand side of the implication (c) holds trivially.
      Finally, (d) holds by (4), (e) holds by (5), and (f) holds by (6).

      \item Case \rname{EVar}. Similar to \rname{ESelf} case.

      \item Case \rname{EAssn}. By assumption we have
      \begin{enumerate}[label=(\arabic*)]
      \item $\config{\cache, \TT, \DT, \DE, x=v, \DS} \rightarrow \config{\cache, \TT, \DT, \DE[x \mapsto v], v, \DS}$
      \item
        \begin{displaymath}
          \infer
          { \TT \vdash \config{\TE, v} \Rightarrow \config{\TE, \tau} }
          {\TT \vdash \config{\TE, x = v} \Rightarrow \config{\TE[x\mapsto \tau], \tau}}
        \end{displaymath}
        by \rname{TAssn} and either \rname{TNil} or \rname{TObject}
      \item $\tau \leq \TS$
      \item $\TE \sim \DE$
      \item $\TT \vdash \TS \sim \DS$
      \item $\cache \sim (\TT, \DT)$
      \end{enumerate}
      Let $\TEalt = \TEalt' = \TE[x\mapsto\tau]$, let $\TS' = \TS$,
      and let $\tau' = \tau$. Notice that in (2),
      the hypothesis can only be proven by either \rname{TNil} or
      \rname{TObject}, both of which are insensitive to the type
      environment. Thus, by the hypothesis of (2), we also have
      $\TT \vdash \config{\TEalt, v} \Rightarrow \config{\TEalt,
        \tau}$, which is (a). Also, (b), (e), and (f) hold trivially by (3), (5), and (6).
      Also, the right-hand side of the implication (c) holds trivially.
      Finally, from (4)
      and the hypothesis of (2) we have $\TEalt = \TE[x\mapsto\tau]
      \sim E[x\mapsto v]$, which is (d).

    \item Case \rname{ENew}. Trivial.

    \item Case \rname{ESeq}. Trivial.


      \item Case \rname{EIfTrue}. By assumption we have
      \begin{enumerate}[label=(\arabic*)]
      \item $\config{\cache, \TT, \DT, \DE, \sif{v}{e_1}{e_2}, \DS} \rightarrow \\\config{\cache, \TT, \DT, \DE, e_1, \DS}$ where $v \neq \snil$
      \item
        \begin{displaymath}
          \infer{\TT \vdash \config{\TE, v} \Rightarrow \config{\TE, \tau} \\\\
            \TT \vdash \config{\TE, e_1} \Rightarrow \config{\TE_1, \tau_1} \\\\
            \TT \vdash \config{\TE, e_2} \Rightarrow \config{\TE_2, \tau_2}}
          {\TT \vdash \config{\TE, \sif{v}{e_1}{e_2}} \Rightarrow \config{\TE_1 \sqcup \TE_2, \tau_1 \sqcup \tau_2}}
        \end{displaymath}
        by \rname{TIf} and \rname{TObject}, since $v\neq\snil$
      \item $\tau_1\sqcup\tau_2 \leq \TS$
      \item $\TE \sim \DE$
      \item $\TT \vdash \TS \sim \DS$
      \item $\cache \sim (\TT, \DT)$
      \end{enumerate}
      Let $\TEalt = \TE$, let $\TEalt' = \TE_1$, let $\TS' = \TS$, and let
      $\tau' = \tau_1$. From the second hypothesis of (2) we trivially
      have (a). Moreover, $\tau' = \tau_1 \leq \tau_1 \sqcup \tau_2$, so by (3)
      we have $\tau' \leq \TS$, which is (b).
      By (4), (5), and (6) we trivially have (d), (e), and (f).
      Finally, by Lemma~\ref{lemma:type-env-leq-union} we have
      $\TEalt' = \TE_1 \leq (\TE_1 \sqcup \TE_2)$, which is the
      right-hand side of the implication (c).

      \item Case \rname{EIfFalse}. Similar to \rname{EIfTrue} case.

      \item Case \rname{EDef}. By assumption we have
        \begin{enumerate}[label=(\arabic*)]
        \item
          $\config{\cache, \TT, \DT, \DE, \sdef{A.m}{\sprem{x}{e}},
            \DS} \rightarrow \\\config{\invcache{\cache}{A.m}, \TT,
            \DT[A.m \mapsto \lambda x.e], \DE, \snil, \DS}$ by \rname{EDef}
        \item $\TT \vdash \config{\TE, \sdef{A.m}{\sprem{x}{e}}} \Rightarrow \config{\TE, \snil}$ by \rname{TDef}
        \item $\snil \leq \TS$
        \item $\TE \sim \DE$
        \item $\TT \vdash \TS \sim \DS$
        \item $\cache \sim (\TT, \DT)$
        \end{enumerate}
        Let $\TEalt = \TEalt' = \TE$, let $\TS' = \TS$, and let
        $\tau' = \snil$.  Then (a) holds trivially by
        \rname{TNil}. (c) holds trivially by definition.  (b), (d), and (e) hold trivially
        by (3), (4), and (5).

        For (f), pick some $B.m' \in
        \dom(\invcache{\cache}{A.m})$.
        Observe that $B.m' \neq A.m$, by construction.  By (6), we
        have
        \begin{enumerate}
        \item $\derivm = (\TT \vdash \config{[y\mapsto\tau_y, \sself\mapsto B], e} \Rightarrow \config{\TE'_y, \tau'_y})$
        \item $\derivt = (\tau'_y \leq \tau_2)$
        \item $\derivm$ and $\derivt$ hold
        \item $\DT(B.m') = \lambda y.e'$
        \item $\TT(B.m') = \tau_y \rightarrow \tau_2$
        \end{enumerate}
        We need to show the above with the same type class table and
        with dynamic class table $\DT[A.m\mapsto \lambda x.e]$. But
        then 1, 2, 3, and 5 are trivial, and since $B.m' \neq A.m$ we
        have $(\DT[A.m\mapsto \lambda x.e])(B.m') = \DT(B.m')$, thus 4
        is trivial.

      \item Case \rname{EType}. This case is very similar to
        \rname{EDef}, except the reduction in the semantics is
        different. By assumption, we have
        \begin{enumerate}[label=(\arabic*)]
        \item
          $\config{\cache, \TT, \DT, \DE, \smethsig{A.m}{\tau_m}, \DS} \rightarrow$\\
          $\config{\upgcache{(\invcache{\cache}{A.m})}{\TT'}, \TT', \DT, \DE, \snil, \DS}$
          where $\TT' = \TT[A.m\mapsto \tau_m]$, by \rname{EType}
        \item $\TT \vdash \config{\TE, \sdef{A.m}{\sprem{x}{e}}} \Rightarrow \config{\TE, \snil}$ by \rname{TDef}
        \item $\snil \leq \TS$
        \item $\TE \sim \DE$
        \item $\TT \vdash \TS \sim \DS$
        \item $\cache \sim (\TT, \DT)$
        \end{enumerate}
        (a)--(d) hold by the same reasoning above.
        To see (e), observe that side condition
        $A.m\not\in\textrm{TApp}(S)$ means that in the typing
        judgments internal to (5), (TApp) is never applied with
        $A.m$. Hence those same judgments hold under $\TT'$, which
        only differs from $\TT$ in its binding or $A.m$.

        Let
        $\cache' = \upgcache{(\invcache{\cache}{A.m})}{\TT'}$.  For
        (f), again pick some $B.m' \in \dom(\cache')$.  Observe that
        $B.m' \neq A.m$, by construction.  By (6), we have
        \begin{enumerate}
        \item $\derivm = (\TT \vdash \config{[y\mapsto\tau_y, \sself\mapsto B], e} \Rightarrow \config{\TE'_y, \tau'_y})$
        \item $\derivt = (\tau'_y \leq \tau_2)$
        \item $\derivm$ and $\derivt$ hold
        \item $\DT(B.m') = \lambda y.e'$
        \item $\TT(B.m') = \tau_y \rightarrow \tau_2$
        \end{enumerate}
        We need to show the above, but in $\cache'$ and with type
        class table $\TT'$ and the same dynamic class table. By
        construction, $\cache'(B.m') = (\derivm', \derivt)$ where
        $\derivm' = (\TT' \vdash \config{[y\mapsto\tau_y,
          \sself\mapsto B], e} \Rightarrow \config{\TE'_y, \tau'_y})$,
        which is 1 and 2.  Notice by construction that $\derivm$
        and $\derivm'$ cannot refer to $A.m$, thus we have 3. Finally,
        4 holds trivially, and 5 holds since $B.m' \neq A.m$ by
        construction.

      \item Case \rname{EAppMiss}. The inductive cases are similar to \rname{EContext}.
        In the non-inductive case, by assumption we have
      \begin{enumerate}[label=(\arabic*)]
      \item $\config{\cache, \TT, \DT, \DE, C[[A].m(v_2)], \DS} \rightarrow$ \\
        $\config{\cache[A.m\mapsto(\derivm,\derivt)], \TT, \DT, [\sself \mapsto [A], x \mapsto v_2], e, (\DE, C)::\DS}$ where
        \begin{enumerate}[label=(1\alph*)]
        \item $\DT(A.m) = \lambda x.e$
        \item $\TT(A.m) = \tau_1 \rightarrow \tau_2$
        \item $\stype{v_2} \leq \tau_1$
        \item $\derivm = (\TT \vdash \config{[x \mapsto \tau_1, \sself \mapsto A], e} \Rightarrow \config{\TE', \tau'_2})$ holds
        \item $\derivt = (\tau'_2 \leq \tau_2)$ holds
        \item $A.m\not\in\dom(\cache)$
        \end{enumerate}
      \item
        \begin{displaymath}
          \infer
          {
            \infer
            {\TT \vdash \config{\TE, [A]} \Rightarrow \config{\TE, A} \\\\
              \TT \vdash \config{\TE, v_2} \Rightarrow \config{\TE, \tau} \\\\
              \TT(A.m) = \tau_1\rightarrow \tau_2 \\
              \tau \leq \tau_1}
            {\TT \vdash \config{\TE, [A].m(v_2)} \Rightarrow \config{\TE, \tau_2}} \\\\
            \vdots
          }
          {\TT \vdash \config{\TE_C, C[[A].m(v_2)]} \Rightarrow \config{\TE'_C, \tau_C}}
        \end{displaymath}
        by \rname{TApp} and \rname{TObject} and possible \rname{TNil}.
      \item $\tau_C \leq \TS$
      \item $\TE_C \sim \DE_C$
      \item $\TT \vdash \TS \sim \DS$
      \item $\cache \sim (\TT, \DT)$
      \end{enumerate}
      Let
      $\TEalt = [x \mapsto \tau_1, \sself \mapsto A]$, let
      $\TEalt' = \TE'$, let
      $\TS' = \tselt{\TE_C}{\tau_2}{\TE'_C}{\tau_C}::\TS$, and let
      $\tau' = \tau'_2$. 
      Then (a) holds immediately by (1d). (b) holds by (1e) and
      construction of $\TS'$. In this case the stack changes, so the
      left-hand side of the implication (c) is false, hence (c) holds
      trivially.
      (d) holds because by \rname{TObject} we have
      $[A]$ has type $A$, and by the second hypothesis of (2), which
      is either \rname{TObject} or \rname{TNil}, we have $v_2$ has
      type $\tau$, and by the last hypothesis of (2) we have
      $\tau \leq \tau_1$.

      Next we show (e).  By (4) we have $\Gamma_C \sim E$, and by (2)
      and the Contextual Substitution Lemma we have
      $\TT \vdash \config{\TE_C[\Box\mapsto\tau_2], C} \Rightarrow
      \config{\TE'_C, \tau_C}$.
      Thus we have $\TT \vdash \tselt{\TE_C}{\tau_2}{\TE'_C}{\tau_C} \sim (E,C)$.
      Further, by (3) we have $\tau_C \leq \TS$. Finally, by (5) we
      have $\TT \vdash \TS \sim \DS$.  Putting this all together, we
      have
      $\TT \vdash \tselt{\TE_C}{\tau_2}{\TE'_C}{\tau_C} \sim (E,C) ::
      \TS \sim (E,C) :: \DS$, which is (e).

      Finally, to show (f), pick some element in the domain of
      $\cache' = \cache[A.m\mapsto(\derivm,\derivt)]$. If we pick some
      $B.m'\neq A.m$ then all the necessary properties hold by (6). If
      we pick $A.m$, then 1 and 2 hold by construction, 3 holds by
      (1d) and (1e), 4 holds by (1a), and 5 holds by (1b).

    \item Case \rname{EAppHit}. This case follows mostly the same
      reasoning as above. The inductive cases are similar to \rname{EContext}.
        In the non-inductive case, by assumption we have
      \begin{enumerate}[label=(\arabic*)]
      \item $\config{\cache, \TT, \DT, \DE, C[[A].m(v_2)], \DS} \rightarrow$ \\
        $\config{\cache, \TT, \DT, [\sself \mapsto [A], x \mapsto v_2], e, (\DE, C)::\DS}$ where
        \begin{enumerate}[label=(1\alph*)]
        \item $\DT(A.m) = \lambda x.e$
        \item $\TT(A.m) = \tau_1 \rightarrow \tau_2$
        \item $\stype{v_2} \leq \tau_1$
        \item $A.m\in\dom(\cache)$
        \end{enumerate}
      \item
        \begin{displaymath}
          \infer
          {
            \infer
            {\TT \vdash \config{\TE, [A]} \Rightarrow \config{\TE, A} \\\\
              \TT \vdash \config{\TE, v_2} \Rightarrow \config{\TE, \tau} \\\\
              \TT(A.m) = \tau_1\rightarrow \tau_2 \\
              \tau \leq \tau_1}
            {\TT \vdash \config{\TE, [A].m(v_2)} \Rightarrow \config{\TE, \tau_2}} \\\\
            \vdots
          }
          {\TT \vdash \config{\TE_C, C[[A].m(v_2)]} \Rightarrow \config{\TE'_C, \tau_C}}
        \end{displaymath}
        by \rname{TApp} and \rname{TObject} and possible \rname{TNil}.
      \item $\tau_C \leq \TS$
      \item $\TE_C \sim \DE_C$
      \item $\TT \vdash \TS \sim \DS$
      \item $\cache \sim (\TT, \DT)$
      \end{enumerate}
      By (6), we have $\cache(A.m) = (\derivm, \derivt)$ where
      $\derivm = (\TT \vdash \config{[x \mapsto \tau_1, \sself \mapsto
        A], e} \Rightarrow \config{\TE', \tau'_2})$
      holds and $\derivt = (\tau'_2 \leq \tau_2)$ holds. Notice that
      we use properties 4 and 5 of the cache in combination with (1a)
      and (1b) to know the assigned types in the cache, and the method
      body at run-time, match in the cached derivation.

      Let
      $\TEalt = [x \mapsto \tau_1, \sself \mapsto A]$, let
      $\TEalt' = \TE'$, let
      $\TS' = \tselt{\TE_C}{\tau_2}{\TE'_C}{\tau_C}::\TS$, and let
      $\tau' = \tau'_2$. 
      Then (a) holds immediately by $\derivm$. (b) holds by $\derivt$ and
      construction of $\TS'$.

      The reasoning for (c)--(e) are the same as the \rname{EAppMiss}
      case. Finally, (f) holds trivially by (6), since the cache did
      not change.

  \item Case \rname{ERet}. We have
      \begin{enumerate}[label=(\arabic*)]
      \item $\config{\cache, \TT, \DT, \DE', v, (\DE, C)::\DS} \rightarrow \\\config{\cache, \TT, \DT, \DE, C[v], \DS}$
      \item $\TT \vdash \config{\TE', v} \Rightarrow \config{\TE', \tau}$ by either \rname{TObject} or \rname{TNil}.
      \item $\tau \leq \tau_C$
      \item $\TE' \sim \DE'$
      \item $\TT \vdash \tselt{\TE_C}{\tau_C}{\TE'_C}{\tau'_C}::\TS \sim (\DE, C)::\DS$
      \item $\cache \sim (\TT, \DT)$
      \end{enumerate}
      Let $\TEalt = \TE_C$, let $\TEalt' = \TE'_C$, and let $\TS' = \TS$.
      By (5), we have
      $\TT \vdash \config{\TEalt[\Box\mapsto\tau_C], C} \Rightarrow
      \config{\TEalt', \tau'_C}$.
      Putting that together with (2) and (3) via the substitution
      lemma, we have
      $\TT \vdash \config{\TEalt, C[v]} \Rightarrow \config{\TEalt',
        \tau''_C}$ where $\tau''_C \leq \tau'_C$. Let $\tau' =
      \tau''_C$, and we have 
      (a).  By (3) we have $\tau'_C \leq \TS$, and since $\tau''_C
      \leq \tau'_C$ we therefore have (b)
      In this case the stack changes, so the
      left-hand side of the implication (c) is false, hence (c) holds
      trivially.
      (d) holds by (5), as does (e). Finally, (f) holds by (6)
    \end{itemize}

\end{proof}

The progress theorem is much simpler:

\begin{theorem}[Progress]
  If
  \begin{enumerate}[label=(\arabic*)]
    \item $\TT \vdash \config{\TE, e} \Rightarrow \config{\TE', \tau}$
    \item $\tau \leq \TS$
    \item $\TE \sim \DE$
    \item $\TT \vdash \TS \sim \DS$
    \item $\cache \sim (\TT, \DT)$
  \end{enumerate}
  then one of the following holds
  \begin{enumerate}
    \item $e$ is a value, or
    \item 
      There exist $\cache'$, $\TT'$, $\DT'$, $\DE'$, $e'$, $\DS'$ such
      that \\
      $\config{\cache, \TT, \DT, \DE, e, \DS} \rightarrow \config{\cache',
        \TT', \DT', \DE', e', \DS'}$, or
  \item $\config{\cache, \TT, \DT, \DE, e, \DS} \rightarrow \blame$
  \end{enumerate}
\end{theorem}

\begin{proof}
  By induction on $e$.
  \begin{itemize}
  \item Case $e = \snil$ or $e = [A]$. These are values, so the theorem holds trivially.
  \item Case $\sself$. By assumption (1) we have
    \begin{displaymath}
      \infer
      { }
      {\TT \vdash \config{\TE, \sself} \Rightarrow \config{\TE, \TE(\sself)}}
  \end{displaymath}
  Thus, $\sself\in\dom(\TE)$. But then by (3),
  $\sself\in\dom(\DE)$. Thus \rname{ESelf} can be applied.

  (Note that assuming we start executing the program in a standard
  environment, $\sself$ will in fact always be bound in all type and
  dynamic environments, unlike variables.)

  \item Case $x$. Similar to $\sself$.

  \item Case $x=v$, $\snew{A}$, $v; e$, $\sdef{A.m}(\sprem{x}{e})$, $\smethsig{A.m}{\tau_m}$. These cases are trivial, as there is one semantics rule for each of these forms, and it will always be able to take a step.
  \item Case $\sif{v}{e_2}{e_3}$. This case is trivial, since either \rname{EIfTrue} or \rname{EIfFalse} will apply.
  \item Case $v_0.m(v_1)$. By assumption (1) we have
    \begin{displaymath}
      \infer
      {\TT \vdash \config{\TE, v_0} \Rightarrow \config{\TE_0, A} \\\\
      \TT \vdash \config{\TE_0, v_1} \Rightarrow \config{\TE_1, \tau} \\\\
      \TT(A.m) = \tau_1\rightarrow \tau_2 \\
      \tau \leq \tau_1 \\
      }
      {\TT \vdash \config{\TE, v_0.m(v_1)} \Rightarrow \config{\TE_1, \tau_2}}
    \end{displaymath}
    There are a few cases. If \rname{EAppNil}, \rname{EAppNExist}, or
    \rname{EAppNTyp} apply, then the theorem holds
    trivially. Otherwise, we must have $v_1 = [A]$ and
    $\DT(A.m) = \sprem{x}{e}$. More importantly, by (1) we have
    $\stype{v_1} \leq \tau_1$, since $\tau = \stype{v_1}$ by (1),
    i.e., $v_1$ has the expected argument type.  Also by (1) we have
    $\TT(A.m) = \tau_1 \rightarrow \tau_2$.

  Now there are two cases. If $A.m\in\dom(\cache)$ we can immediately
  apply \rname{EAppHit}. Otherwise, if $A.m\not\in\dom(\cache)$, then
  we must have
  $$\derivm = \left(\TT \vdash \config{[x \mapsto \tau_1, \sself
      \mapsto A], e} \Rightarrow \config{\TE', \tau}\right)$$
  holds and $\derivt = (\tau \leq \tau_2)$ holds because
  \rname{EAppNTyp} did not apply.  But then combining this with our
  previous assumptions we can apply \rname{EAppMiss}.

  \item Else $e = C[e']$. Holds by induction and \rname{EContext}.
  \end{itemize}
\end{proof}

Finally, we can put these together to prove soundness.

\begin{theorem}[Soundness]
  If
  $\emptyset \vdash \config{\emptyset, e} \Rightarrow \config{\TE', \tau}$
  then either $e$ reduces to a value, $e$ reduces to $\blame$, or $e$
  does not terminate.
\end{theorem}

\begin{proof}
  Let $\cache = \emptyset$, let $\TT = \emptyset$, let
  $\TE = \emptyset$, let $\DE = \emptyset$, let $\DT = \emptyset$, let
  $\DS = (\emptyset, \Box)::\cdot$, and let
  $\TS = \tselt{\emptyset}{\tau}{\emptyset}{\tau}$.
  Then by assumption we have $\TT \vdash \config{\TE, e} \Rightarrow
  \config{\TE', \tau}$. By construction we have $\tau \leq \TS$ and
  $\TE \sim \DE$ and $\TT \vdash \TS \sim S$ and $\cache \sim (\TT,
  \DT)$.
  Thus, these choices of $\cache$, $\TT$, $\TE$, $\DT$, $\DS$, and
  $\TS$ satisfy the preconditions of progress and preservation.
  Thus soundness holds by standard arguments.
\end{proof}

\end{document}